\documentclass{amsart}
             
\usepackage{amssymb}
\usepackage{amsthm}                
\usepackage{amstext}
\usepackage{amsbsy}
\usepackage{amscd}
\usepackage{enumerate}
\usepackage{subfigure}
\usepackage[noadjust]{cite}
\usepackage{hyperref}
\usepackage{graphicx} 
\usepackage{mathtools}
\usepackage{pict2e}
\usepackage[usenames,dvipsnames]{color}
\usepackage{appendix}

\usepackage{tikz}
\usetikzlibrary{matrix}

\newtheorem{thm}{Theorem}[section]
\newtheorem{cor}[thm]{Corollary}
\newtheorem{lem}[thm]{Lemma}
\newtheorem{prop}[thm]{Proposition}
\newtheorem{conj}[thm]{Conjecture}

\theoremstyle{remark}
\newtheorem{rem}[thm]{Remark}

\theoremstyle{definition}
\newtheorem{defn}[thm]{Definition}

\theoremstyle{plain}

\newcommand{\norm}[1]{\left\Vert#1\right\Vert}
\newcommand{\abs}[1]{\left\vert#1\right\vert}
\newcommand{\set}[1]{\left\{#1\right\}}
\newcommand{\R}{\mathbb{R}}
\newcommand{\N}{\mathbb{N}}
\newcommand{\Z}{\mathbb{Z}}
\newcommand{\vt}[1]{\mathbf{#1}}
\newcommand{\diff}{\mathrm{Diff}}
\newcommand{\emb}{\mathrm{Emb}}
\newcommand{\Int}{\mathrm{Int}}
\newcommand{\dist}{\mathbf{dist}}
\newcommand{\proj}[2]{\mathbf{Proj}_{#1}\left(#2\right)}
\renewcommand{\dim}{\mathrm{dim}}
\newcommand{\loc}{\mathrm{loc}}
\newcommand{\Tr}{\mathrm{Tr}}
\newcommand{\hdim}{\mathbf{dim}_\mathrm{H}}
\newcommand{\lhdim}{\mathbf{dim}_\mathrm{H}^\mathrm{loc}}
\newcommand{\hdist}{\mathbf{dist}_\mathrm{H}}
\newcommand{\diam}{\mathbf{diam}}

\newcommand{\cites}{\cite}

\begin{document}

\title[Spectrum of Quantum Ising Quasicrystal]{On the spectrum of 1D quantum Ising quasicrystal}

\author[W. N. Yessen]{William N. Yessen}
\email{wyessen@math.uci.edu}
\address{Department of Mathematics, UC Irvine, Irvine, CA 92697}
\thanks{The author was supported by the NSF grants DMS-0901627, PI: A. Gorodetski and IIS-1018433, PI: M. Welling and Co-PI: A. Gorodetski}


\date{\today}

\begin{abstract}

We consider one dimensional quantum Ising spin-$1/2$ chains with two-valued nearest neighbor couplings arranged in a quasi-periodic sequence, with uniform, transverse magnetic field. By employing the Jordan-Wigner transformation of the spin operators to spinless fermions, the energy spectrum can be computed exactly on a finite lattice. By employing the transfer matrix technique and investigating the dynamics of the corresponding trace map, we show that in the thermodynamic limit the energy spectrum is a Cantor set of zero Lebesgue measure. Moreover, we show that local Hausdorff dimension is continuous and non-constant over the spectrum. This forms a rigorous counterpart of numerous numerical studies.

\end{abstract}

\maketitle

\section{Introduction}

Since the discovery of quasicrystals \cites{Levine1984,Levine1986, Shechtman1984, Socolar1986}, quasi-periodic models in mathematical physics have formed an active area of research. In all of these models, quasi-periodicity is introduced via so-called quasi-periodic sequences, which, roughly speaking, are somewhat intermediate between periodic and random (for a textbook exposition, see \cite{Queffelec2010, Fogg2002}), and are meant to model microscopic quasi-periodic structure of quasicrystals (not all such sequences necessarily correspond to a microscopic organization of an actual physical material, but the Fibonacci substitution sequence, which we consider here, does correspond to an actual quasicrystal made up of two microscopic constituents whose arrangement resembles the Fibonacci substitution sequence). One of the main tools in the investigation of such models has been renormalization analysis, leading to renormalization maps whose properties (or, better to say, their action on appropriate spaces, 
typically $\R^n$) yield strong implications for the underlying models, which are usually very difficult or impossible to obtain by other means. These renormalization maps have been called \textit{trace maps} (roughly) since they were originally introduced in \cites{Kadanoff0000, Kohmoto1983, Ostlund1983} (see also \cites{Horowitz1972, Southcott1979, Traina1980, Jorgensen1982, Kohmoto1992, Baake1999, Roberts1994b, Roberts1994} and references therein). The terminology is meant to emphasize their intimate connection to the transfer matrix formalism, perhaps better known to statistical physicists, which is a renormalization procedure that allows one to study some relevant properties of statistical-mechanical models via traces of appropriate \textit{transfer operators}. (One has to note, however, that in the past few decades these techniques have been greatly generalized by mathematicians, leading to exciting results in a few fields; among the wider known ones would be spectral theory [in particular of ergodic 
Schr\"odinger operators] and dynamical systems [in the spirit of works by R. Bowen, D. Ruelle and Ya. Sinai on thermodynamic formalism]). The method of trace maps has led, for example, to fundamental results in spectral theory of discrete Schr\"odinger operators and Ising models on one-dimensional quasi-periodic lattices (for Schr\"odinger operators: \cites{Kohmoto1983, Casdagli1986, Suto1987, Damanik2009, Damanik2010, Damanik2000, Damanik2008, Damanik2005, Bellissard1989, Raymond1997, Simone2009}, for Ising models: \cites{Tsunetsugu1987, Benza1989, Ceccatto1989, Hermisson1997, Doria1988, Benza1990, You1990, Tong1997}, and 
references therein).

Quasi-periodicity and the associated trace map formalism is still an area of active investigation, mostly in connection with their applications in physics. In this paper we use these techniques to investigate the energy spectrum of one-dimensional quantum Ising spin chains with two-valued nearest neighbor couplings arranged in a quasi-periodic sequence, with uniform, transverse magnetic field. We shall concentrate on the quasi-periodic sequence generated by the Fibonacci substitution on two symbols, which (probably due to the original choice of models in early 1980's) is the most widely studied case (and a representative example of many observed phenomena exhibited by quasi-periodic models in general).

One-dimensional quasi-periodic quantum Ising spin chains have been investigated (analytically and numerically) over the past two decades \cites{Benza1989, Ceccatto1989, Hermisson1997, Baake1999,Doria1988, Benza1990, You1990, Igloi1988, Turban1994,Igloi1997, Igloi1998, Igloi2007}. Numerical and some analytic results suggest Cantor structure of the energy spectrum, with nonuniform local scaling (i.e. a multifractal). The multifractal structure of the energy spectrum has not been shown rigorously.

Our aim here is to prove the multifractality of the energy spectrum and investigate its fractal dimensions. We achieve this by carrying out a renormalization procedure, using in a crucial way quasi-periodicity of (i.e. the repetitive, albeit not periodic, nature of) the sequence of interaction couplings (which are chosen according to the aforementioned Fibonacci substitution). The renormalization map turns out to be a second degree polynomial acting on $\R^3$ -- the so-called \textit{trace map}. We then investigate the discrete time dynamics of this polynomial map, relating the energy spectrum to invariant sets for the renormalization map. Techniques from hyperbolic and partially hyperbolic dynamics are then employed to investigate topological and fractal-dimensional properties of these invariant sets, concluding with our main result about the energy spectrum which, roughly speaking, states: \textit{The spectrum is a zero measure Cantor set with non-constant local scaling (i.e. a multifractal); the spectrum 
and its fractal dimensions depend continuously on the parameters of the model}. We should note that these observations are not entirely in line with what has been observed in spectral analysis of quasi-periodic Schr\"odinger operators, since in the latter, the spectrum, while a Cantor set, is not a multifractal (i.e. local scaling is uniform throughout the spectrum). We comment on this in more detail in Section \ref{sec:cc}. 

Certainly the technique described above is not new and has been employed numerous times (see the references above), most notably in the context of quasi-periodic Schr\"odinger operators in one dimension; however, to the best of the author's knowledge, this is the first rigorous treatment of quantum one-dimensional discrete quasi-periodic Ising models that provides rigorous results regarding the multifractal structure of the spectrum. Even though, as is mentioned later in the course of this paper, the Ising model that we are concerned with here can be mapped canonically to a Jacobi operator (which resembles in many ways the Schr\"odinger operator), the method of trace maps, while still applicable, presents certain technical difficulties that are not present in the context of Schr\"odinger operators. We discuss in more detail these difficulties in Section \ref{sec:cc}; let us only mention here that the techniques that we developed in this paper seem to be applicable in a wide range of models, and can be 
substantially generalized and presented in a model-independent fashion. The present paper, however, is already rather involved and technical (which is to be expected, as one has to perform a number of transformations from the original spectral-theoretic problem to a problem of geometry and dynamical systems); for this reason it had been decided to postpone generalizations (at the time when this paper was written). First steps towards developing a general toolbox were later taken in \cite[Section 2]{Damanik0000my1} jointly with D. Damanik and P. Munger. In that same paper, \cite{Damanik0000my1}, we investigated spectral properties of a class of five-diagonal unitary operators, commonly called the \textit{CMV matrices}, that play a central role in the theory of orthogonal polynomials on the unit circle (in fact, CMV matrices are to orthogonal polynomials on the unit circle as the tridiagonal Jacobi operators are to orthogonal polynomials on the real line \cite{Simon2005a, Simon2005b}). Not surprisingly, the 
techniques that we develop here (and 
generalize in \cite{Damanik0000my1}), among other techniques, were successfully applied to spectral theory of quasi-periodic CMV matrices in \cite{Damanik0000my1}. In the same fashion, these techniques have been successfully applied to quasi-periodic Jacobi matrices in \cite{Yessen2011a}. Applications of our results from \cite{Damanik0000my1} will be appearing in our forthcoming paper \cite{Damanik0000my2}; among the applications are the quantum walks on one-dimensional lattices with quasi-periodic (Fibonacci in particular) coin flips, and a canonical transformation between classical one-dimensional Ising models with external magnetic field and CMV matrices. 

Not to get sidetracked too far, let us conclude this introduction by noting that we study and apply the trace map as a \textit{real analytic map}; study of the complexified version (though within a different context) was carried out by S. Cantat in \cite{Cantat2009}. We have applied the trace map as \textit{a holomorphic map on $\mathbb{C}^3$} in \cite{Yessen2012a} as a renormalization map for the classical one-dimensional Ising model with quasi-periodic nearest neighbor interaction and magnetic field, and were able to relate the analyticity of the free energy function to the analyticity of the escape rate of orbits under the action of the trace map, proving absence of phase transitions. Also in \cite{Yessen2012a} we applied the techniques from the present paper to obtain precise description of the Lee-Yang zeros of the classical model in the thermodynamic limit. (The results in \cite{Yessen2012a} were obtained some months after the present paper was written). At this point we would like to mention the work 
of P. Bleher et. al. \cite{Bleher1992, Bleher200Xa, Bleher200Xb} on Ising models on certain two-dimensional lattices (though not quasi-periodic and not involving trace maps), where the action of the renormalization map was treated as a holomorphic dynamical system. 

Let us also briefly mention that, apart from their applications, trace maps present an interest to the dynamical systems community as a family of polynomial maps exhibiting quite rich dynamics (\cite{Roberts1994, Roberts1994b, Roberts1996, Cantat2009, Humphries2007} and references therein, and our forthcoming work \cite{Gorodetski200X}). In addition to the references from above, for a broad overview of the recent developments and open problems, the interested reader may consult the surveys \cite{Damanik200X} (with emphasis on the Schr\"odinger operator) and the forthcoming \cite{Yessen200X} (with emphasis on the dynamics of trace maps and applications to a class of models, including quantum and classical Ising models).

\section{The 1D quasi-periodic quantum Ising chain}\label{part1}

For a general overview of quasi-periodic (including Fibonacci) Ising models, see, for example, \cite{Grimm1995}.

\subsection{The model}\label{part1_1}

Let $J_0, J_1 > 0$. Construct a $\set{J_0,J_1}$-valued sequence $\set{\widetilde{J}_n}_{n\in\N}$ by applying repeatedly the Fibonacci substitution rule on two letters:
\begin{align}\label{eq:sub}
J_0\mapsto J_0J_1 \text{\hspace{5mm}and\hspace{5mm}} J_1\mapsto J_0,
\end{align}
starting with $J_0$:
\begin{align*}
J_0\mapsto J_0J_1\mapsto J_0J_1J_0\mapsto J_0J_1J_0J_0J_1\mapsto\cdots,
\end{align*}
at each step substituting for $J_0$ and $J_1$ according to the substitution rule \eqref{eq:sub}. By this procedure an infinite sequence is constructed, which we call $\set{\widetilde{J}_n}_{n\in\N}$ (see \cite{Queffelec2010} for more details on substitution sequences).

Let $\widehat{J}_k$ be the finite word after $k$ applications of the substitution rule. It is easy to see that the following recurrence relation holds:
\begin{align}\label{model_eq0}
\widehat{J}_{k+1} = \widehat{J}_{k}\widehat{J}_{k-1}.
\end{align}
The word $\widehat{J}_k$ has length $F_k$, where $F_k$ is the $k$th Fibonacci number. The quasi-periodic (Fibonacci in our case) one-dimensional quantum Ising model on the finite one-dimensional lattice of $N$ nodes with transversal external field is given by the 
\textit{Ising Hamiltonian}
\begin{align*}
\mathcal{H} = -\sum_{n=1}^N \widetilde{J}_n\sigma_{n}^{(x)}\sigma_{n+1}^{(x)} - h\sum_{n=1}^N \sigma_n^{(z)},
\end{align*}
where $h > 0$ is the external magnetic field in the direction transversal to the lattice. The matrices $\sigma_j^{(x),(z)}$ are spin-$1/2$ operators in the $x$ and $z$ directions, respectively, given by
\begin{align*}
\sigma_n^{(x),(z)} = \underbrace{\mathbb{I}\oplus\cdots\oplus\mathbb{I}}_{n-1\text{ times }}
\oplus\sigma^{(x),(z)}
\oplus\underbrace{\mathbb{I}\oplus\cdots\oplus\mathbb{I}}_{N-n+1\text{ times}},
\end{align*}
where $\mathbb{I}$ is the $2\times 2$ identity matrix. Here $\sigma^{(x),(z)}$ are the Pauli matrices given by
\begin{align}\label{eq:pauli-mat}
\sigma^{(x)} = 
\begin{pmatrix}
0 & 1\\
1 & 0
\end{pmatrix}
\text{\hspace{5mm}and\hspace{5mm}}
\sigma^{(z)} = 
\begin{pmatrix}
1 & 0\\
0 & -1
\end{pmatrix}.
\end{align}
The magnetic field $h$ can be absorbed into interaction strength couplings $J_0, J_1$, so $\mathcal{H}$ can be rewritten as
\begin{align}\label{model_eq1}
\mathcal{H}= -\sum_{n=1}^N \left(\widetilde{J}_n\sigma_n^{(x)}\sigma_{n+1}^{(x)} + \sigma_n^{(z)}\right).
\end{align}
The Hamiltonian $\mathcal{H}$ in \eqref{model_eq1} acts on $\mathbb{C}^{2N}$, where we assume periodic boundary conditions: $\sigma_{N+1}^{(x),(z)} = \sigma_0^{(x),(z)}$. Here $\sigma_j^{(\alpha)}$, $\alpha\in\set{x,z}$, acts on the finite sequence $(c_1,\dots,c_N)\in\mathbb{C}^{2N}$ by acting by $\sigma^{(\alpha)}$  from \eqref{eq:pauli-mat} on the $j$-th entry, while leaving the other entries unchanged.

\subsection{Fermionic representation}\label{part1_2}

The spin model from the previous section can in fact be attacked as a so-called free-fermion model by performing the so-called Jordan-Wigner transformation, that transforms the Pauli operators $\sigma^{(x), (y)}$ into anti-commuting Fermi creation and annihilation operators. This technique dates at least as far back as the classical paper by P. Jordan and E. Wigner on second quantization \cite{Jordan1928}. The advantage in performing this transformation, is that the resulting Hamiltonian, in terms of the Fermi operators, can be diagonalized due to the anticommuting property of Fermi operators (commonly known in the physics community as the canonical commutation relations, or CCR for short). Furthermore, the Jordan-Wigner transformation is canonical (to conform to standard physics terminology) in the sense that it is invertible (in particular, no information is introduced and no information is lost by performing this transformation).

For convenience, let us denote by $\mathcal{H}_k$ the Hamiltonian in \eqref{model_eq1} on a lattice of size $F_k$. We can then extend $\mathcal{H}_k$ periodically to a Hamiltonian $\widehat{\mathcal{H}}_k$ over the periodic infinite lattice with the unit cell of length $F_k$. The operator $\mathcal{H}_k$, and hence also $\widehat{\mathcal{H}}_k$, can be cast into fermionic representation by means of the Jordan-Wigner transformation:
\begin{align}\label{model_eq2}
\mathcal{H}_k = \sum_{i,j}\left[c_i^\dagger A_{ij}c_j + \frac{1}{2}\left(c_i^\dagger B_{ij}c_j^\dagger + (c_i^\dagger B_{ij}c_j^\dagger)^\dagger\right)\right],
\end{align}
where $c_i$, $1\leq i \leq F_k$, are anticommuting fermionic operators and $a^\dagger$ denotes Hermitian conjugation of $a$. The terms $A_{ij}$ and $B_{ij}$ that appear in \eqref{model_eq2} are entries of the matrices $A, B$ given by
\begin{align*}
A_{ii}& = -2,& A_{i,i+1}= A_{i+1,i}& = -\widetilde{J}_i,& A_{1, F_k} &= -\widetilde{J}_{F_k};\\
B_{i,i+1}& = -\widetilde{J}_i,& B_{i+1,i}& = \widetilde{J}_i,& B_{1, F_k}& = \widetilde{J}_{F_k};&
\end{align*}
all other entries being zero. This method is due to E. Lieb et. al. \cite{Lieb1961}, and its specialization to the Hamiltonian \eqref{model_eq1} is presented in some detail in \cite{Doria1988}. We do not go into any further details here and invite the reader to consult the mentioned works. 

Now, after we have performed the Jordan-Wigner transformation, the energy spectrum of the Hamiltonian in \eqref{model_eq2} can be computed by solving the so-called $c$-numeric equation for $\lambda$:
\begin{align}\label{eq:matrices}
(A - B)\phi& = \lambda\psi,&\\
(A+B)\psi& = \lambda\phi&\notag
\end{align}
(see \cite[Section B]{Lieb1961}). This equation can be written in the form \cites{Benza1989, Ceccatto1989, You1990}
\begin{align*}
\Phi_{i + 1} = M_i(\lambda)\Phi_i,
\end{align*}
where $\Phi_i = (\psi_i, \phi_i)^T$, and
\begin{align*}
M_i(\lambda) = \begin{pmatrix}
-{1}/{\widetilde{J}_i} & {\lambda}/{2\widetilde{J}_i}\\
-{\lambda}/{2\widetilde{J}_i} & {(\lambda^2 - 4\widetilde{J}_i^2)}/{4\widetilde{J}_i}
\end{pmatrix}.
\end{align*}
Thus the wave function $\Phi$ at site $N$ is given by
\begin{align*}
\Phi_{N+1} = M_{N}(\lambda)\times M_{N-1}(\lambda)\times\cdots\times M_0(\lambda)\Phi_0.
\end{align*}
Letting $\widehat{M}_k(\lambda)$ denote the transfer matrix over $F_k$ sites, using the recurrence relation in \eqref{model_eq0}, we obtain
\begin{align}\label{model_eq3}
\widehat{M}_{k+1}(\lambda) = \widehat{M}_{k}(\lambda)\times\widehat{M}_{k-1}(\lambda),
\end{align}
for $k\geq 2$.

Returning to the Hamiltonian $\widehat{\mathcal{H}}_k$, we see that the wave function at site $nF_k$ is given by
\begin{align*}
\Phi_{nF_k} = \widehat{M}_k(\lambda)^n \Phi_0.
\end{align*}
The wave function over the infinite lattice should not diverge exponentially. Hence we allow only those values of $\lambda$ for which the eigenvalues of $\widehat{M}_k(\lambda)$ lie in $[-1, 1]$. Since $\widehat{M}_k(\lambda)$ is unimodular, this is equivalent to the requirement
\begin{align}\label{eq:trace-bound}
\frac{1}{2}\abs{\Tr\widehat{M}_k(\lambda)}\leq 1.
\end{align}
Let
\begin{align*}
x_k(\lambda) = \frac{1}{2}\Tr\widehat{M}_k(\lambda).
\end{align*}
Using the recursion relation \eqref{model_eq3}, one may derive the recursion relation on the traces given by (see \cites{Kadanoff0000, Kohmoto1983, Ostlund1983} and, for a more general discussion, \cite{Roberts1994b})
\begin{align*}
x_{k+1} = 2x_kx_{k-1} - x_{k-2}.
\end{align*}
Thus, in accordance with \eqref{eq:trace-bound}, we require that $\abs{x_k}\leq 1$. Define the so-called \textit{Fibonacci trace map} $f: \R^3\rightarrow\R^3$ by
\begin{align}\label{model_eq4}
f(x,y,z) = (2xy - z, x, y).
\end{align}
Set
\begin{align*}
M_{-1}(\lambda)& = \begin{pmatrix}
{J_0}/{J_1} & \lambda(J_1^2 - J_0^2)/2J_0J_1\\
0 & {J_1}/{J_0}
\end{pmatrix},\\
M_0(\lambda)& = \begin{pmatrix}
-{1}/{J_0} & {\lambda}/{2J_0}\\
-{\lambda}/{2J_0} & {(\lambda^2 - 4J_0^2)}/{4J_0}
\end{pmatrix},
\text{\hspace{5mm}}\\
M_1(\lambda)& = \begin{pmatrix}
-{1}/{J_1} & {\lambda}/{2J_1}\\
-{\lambda}/{2J_1} & {(\lambda^2 - 4J_1^2)}/{4J_1}
\end{pmatrix} = M_{0}(\lambda)\times M_{-1}(\lambda).
\end{align*}
Then
\begin{align*}
x_{-1}(\lambda) = \left(\frac{J_0}{J_1} + \frac{J_1}{J_0}\right)/2,
\text{\hspace{5mm}}
x_0(\lambda) = \frac{\lambda^2 - (4 + 4J_0^2)}{8J_0},
\text{\hspace{5mm}}
x_1(\lambda) = \frac{\lambda^2 - (4 + 4J_1^2)}{8J_1}.
\end{align*}
It is convenient to absorb the factor $1/4$ into $\lambda^2$ and rewrite
\begin{align*}
x_0(\lambda) = \frac{\lambda^2 - (1 + J_0^2)}{2J_0},
\text{\hspace{5mm}}
x_1(\lambda) = \frac{\lambda^2 - (1 + J_1^2)}{2J_1}.
\end{align*}
Define the line of initial conditions $\gamma_{(J_0,J_1)}:(-\infty,\infty)\rightarrow\R^3$ by
\begin{align}\label{model_eq5}
\gamma_{(J_0, J_1)}(\lambda) = \left(\frac{\lambda^2 - (1 + J_1^2)}{2J_1}, \frac{\lambda^2 - (1 + J_0^2)}{2J_0}, \left(\frac{J_0}{J_1} + \frac{J_1}{J_0}\right)/2\right).
\end{align}
Let $\pi$ denote projection onto the third coordinate, and define
\begin{align}\label{model_eq6}
\sigma_k(J_0,J_1) = \set{\lambda: \abs{\pi\circ f^{k}\circ \gamma_{(J_0,J_1)}(\lambda)}\leq 1},
\end{align}
where $f^k$ denotes $k$-fold composition
\begin{align*}
f^k = \underbrace{f\circ f\circ\cdots\circ f}_{k \text{ times }},\hspace{4mm}k\geq 0.
\end{align*}
For the sake of simplifying notation, we shall write simply $\sigma_k$, keeping in mind its implicit dependence on the choice of $J_0$ and $J_1$. \textit{With this setup, the set $\sigma_k$ is the excitation spectrum of the periodic free-fermion model or, equivalently (via the inverse Jordan-Wigner transformation), the energy spectrum of the periodic spin model}. We are interested in understanding the spectrum in the thermodynamic limit, that is, $k\rightarrow\infty$.

\subsection{The problem and main results}\label{part1_3}

We wish to investigate the energy spectrum of $\widehat{\mathcal{H}}_k$ in the thermodynamic limit (that is, $k\rightarrow\infty$). Since $\pi\circ f^{k}\circ\gamma_{(J_0,J_1)}(\lambda)$ is a polynomial in $\lambda$, $\sigma_k$ is a union of finitely many compact intervals (in general, see the exposition on Floquet theory  applicable to periodic Hamiltonians, in, for example, \cite[Chapter 4]{Toda1981}). Supported by numerical evidence, it is believed that as $k\rightarrow\infty$, the sequence $\set{\sigma_k}_{k\in\N}$ shrinks to a Cantor set \cites{Benza1990, Ceccatto1989, Doria1988, You1990} (i.e. a nonempty, compact, totally disconnected set with no isolated points). In Theorem \ref{thm:main} below we make precise the notion of the energy spectrum in the thermodynamic limit, and we rigorously examine its multifractal nature. Before we continue, however, we need to set up some notation.

We denote the Hausdorff metric on $\mathcal{P}(\R)$ by $\hdist$:
\begin{align*}
\hdist(A, B) = \max\set{\adjustlimits\sup_{a\in A}\inf_{b\in B}\set{|a -b|}, \adjustlimits\sup_{b\in B}\inf_{a\in A}\set{|a - b|}}.
\end{align*}
We denote the Hausdorff dimension of a set $A$ by $\hdim(A)$, and the local Hausdorff dimension of $A$ at a point $a\in A$ by $\lhdim(A, a)$:
\begin{align*}
\lhdim(A, a) = \lim_{\epsilon\rightarrow 0^+}\hdim\left((a-\epsilon, a+\epsilon)\cap A\right).
\end{align*}

We are now ready to state the main result of this paper. 

\begin{thm}\label{thm:main}
 
Fix $J_1 > 0$. There exists $r_0 = r_0(J_1)\in (0,1)$ such that for all $J_0$ satisfying $J_1/J_0\in(1 - r_0, 1+r_0)$, $J_0\neq J_1$, the following statements hold.

\begin{enumerate}[i.]
\item There exists a compact nonempty set $B_\infty(J_0,J_1)\subset \R$ such that $\sigma_k\xrightarrow[k\rightarrow\infty]{}B_\infty$ in the Hausdorff metric;
\item $B_\infty(J_0,J_1)$ is a Cantor set;
\item $\lhdim(B_\infty(J_0, J_1), b)$ depends continuously on $b\in B_\infty(J_0, J_1)$, is non-constant and is strictly between zero and one; consequently $\hdim(B_\infty(J_0, J_1))\in (0, 1)$, and therefore the Lebesgue measure of $B_\infty(J_0, J_1)$ is zero;
\item $\hdim(B_\infty(J_0,J_1))$ is continuous in the parameters $(J_0,J_1)$.
\end{enumerate}
\end{thm}
Convergence of the sequence $\set{\sigma_k}$ to a (nonempty, compact) limit and multifractal nature of this limit (statements (i) and (ii) of the theorem) was observed numerically in \cites{Doria1988, You1990, You1990a}. 

We should add that we believe the restrictions on $J_0, J_1$ in the statement of Theorem \ref{thm:main} are not necessary; however, our present techniques do not extend to the general case (i.e. to cover all values of $J_0, J_1$). We record our belief here formally as a conjecture:

\begin{conj} The conclusion of Theorem \ref{thm:main} holds for all $J_0, J_1 > 0$. 
\end{conj}

We should remark here that the Ising Hamiltonian above is equivalent, via a unitary transformation, to the tight binding model:
\begin{align}\label{eq:tbm}
(\mathcal{T}\theta)_n = \widetilde{J}_{n-1}\theta_{n-1} + (1 + \widetilde{J}_n^2)\theta_n + \widetilde{J}_{n+1}\theta_{n+1}.
\end{align}
In fact, it can be seen easily that solving \eqref{eq:matrices} is equivalent to solving the equation
\begin{align*}
(A + B)(A - B)\phi = \lambda^2\phi.
\end{align*}
This, on the other hand, is equivalent to solving
\begin{align*}
\mathcal{T}\phi = \lambda^2\phi.
\end{align*}
Now, $\mathcal{T}$ is a member of a family of tridiagonal Hamiltonians investigated in \cite{Yessen2011a} (results of \cite{Yessen2011a} were obtained some months after a preprint of the present paper appeared). From our results in \cite{Yessen2011a}, combined with Proposition \ref{part2_cor10} from Section \ref{sec:proof-main-i} below, it follows that with $\mathcal{T}$ from \eqref{eq:tbm}, we have $\sigma(\mathcal{T})\subset \R^+$ and
\begin{align*}
B_\infty = \set{\pm\sqrt{\lambda}: \lambda\in\sigma(\mathcal{T})},
\end{align*}
where $\sigma(\mathcal{T})$ is the spectrum of $\mathcal{T}$.

After $B_\infty$ has been tied to the spectrum of $\mathcal{T}$ as above, results of \cite{Avron1990} yield a proof of statement (i) of Theorem \ref{thm:main}. Moreover, in this case the restrictions given in the statement of Theorem \ref{thm:main} are not necessary; these restrictions are, however, still a necessity to our techniques for proving the remaining statements (ii)--(iv). Since previous numerical methods have relied entirely on the dynamics of the Fibonacci trace map, motivated by providing a rigorous counterpart, we provide an alternative proof of Theorem \ref{thm:main}-i in Section \ref{sec:proof-main-i}, based entirely on dynamical properties of the trace map.

On another note, we have the following remark regarding the connection with the Jacobi operator $\mathcal{T}$.

\begin{rem}
 In \cite{Yessen2011a} it is proved (using the techniques that we develop here) that the spectrum of $\mathcal{T}$ is a Cantor (in fact, multifractal) set of zero Lebesgue measure; moreover, we proved (heavily relying on the previous work for Schr\"odinger operators) that the spectral measures are purely singular continuous. The following question was justly raised by one of the referees of this article: \textit{What implications would pure singular continuity of the spectral measures have in the context of the Ising model?} While we have not investigated this question in detail, we postulate a connection with spin-spin correlation decay, and point to \cite{Hermisson1997} and references therein for some work in this direction.
\end{rem}

\section{Hierarchy of results}\label{sec:strategy}

Given that the proof of Theorem \ref{thm:main} is quite technical, we provide below a diagram of lemmas, propositions and theorems, demonstrating their hierarchy, in hopes of making navigation through the technical passages of the paper easier (\textbf{T} stands for Theorem, \textbf{P} stands for Proposition, \textbf{L} stands for Lemma, and \textbf{R} stands for Remark).

\begin{center}
\begin{tikzpicture}
 \matrix(m)[matrix of nodes, row sep = 1em,
	    column sep = 1em, minimum width = 1em]
  {
    & \textbf{T} \ref{thm:main} (i) & & & \textbf{T} \ref{thm:main} (ii) &\\
    & \textbf{P} \ref{part2_cor10} & & & \textbf{P} \ref{part2_cor9} & \textbf{P} \ref{part2_prop11} \\
    \textbf{L} \ref{part2_lem0_1} & \textbf{L} \ref{part2_cor10-helper1} & \textbf{L} \ref{lem:thm-i-helper1} & \textbf{L} \ref{part2_lem0_1} & \textbf{R} \ref{part2_rem5} & \textbf{P} \ref{part2_prop7}\\
    & \textbf{P} \ref{part2_cor9} & & & & \textbf{T} \ref{part2_thm1}\\
    & \textbf{T} \ref{thm:main} (iii) & & & \textbf{T} \ref{thm:main}\\
    \textbf{P} \ref{cantat_prop} & \textbf{P} \ref{damanik_gorodetski_prop} & \textbf{L} \ref{holder_lemma} & & \textbf{P} \ref{damanik_gorodetski_prop}\\
    & \textbf{P} \ref{lip_cont} & \textbf{P} \ref{part2_prop11} & \textbf{P} \ref{palis_viana} & & \\
  };
  \path[-stealth]
		(m-2-2) edge node [] {} (m-1-2)
		(m-3-1) edge node [] {} (m-2-2)
		(m-3-2) edge node [] {} (m-2-2)
		(m-3-3) edge node [] {} (m-2-2)
		(m-4-2) edge node [] {} (m-3-2)
		(m-2-5) edge node [] {} (m-1-5)
		(m-2-6) edge node [] {} (m-1-5)
		(m-3-4) edge node [] {} (m-2-5)
		(m-3-5) edge node [] {} (m-2-5)
		(m-3-6) edge node [] {} (m-2-5)
		(m-4-6) edge node [] {} (m-3-6)
		(m-6-1) edge node [] {} (m-5-2)
		(m-6-2) edge node [] {} (m-5-2)
		(m-6-3) edge node [] {} (m-5-2)
		(m-6-5) edge node [] {} (m-5-5)
		(m-7-2) edge node [] {} (m-6-3)
		(m-7-3) edge node [] {} (m-6-3)
		(m-7-4) edge node [] {} (m-6-3);
\end{tikzpicture}
\end{center}

In what follows, the least transparent proof of an auxiliary result is that of Proposition \ref{part2_prop11}. Below is a diagram illustrating its dependence on some technical lemmas. 

\begin{center}
\begin{tikzpicture}
 \matrix(m)[matrix of nodes, row sep = 1em,
	    column sep = 1em, minimum width = 1em]
	{
	& & \textbf{P} \ref{part2_prop11} & & &\\
	& & \textbf{L} \ref{part2_lem12} & & &\\
	\textbf{L} \ref{cone_projection} & \textbf{C} \ref{invariance_and_expansion} & \textbf{C} \ref{cor_cone_invariance} & \textbf{P} \ref{part2_prop0} & \textbf{L} \ref{lem_arb1} & \textbf{L} \ref{bound_on_v_to_lambda} \\
	\textbf{L} \ref{cones_on_compacta} & \textbf{L} \ref{invariance_of_cones} & \textbf{L} \ref{uniform_expansion} & \textbf{L} \ref{invariance_3d_cones} & & \\
	& &  \textbf{L} \ref{bound_on_gradiants} & \textbf{L} \ref{technical} & \textbf{L} \ref{using_technical} &\\
	};
	\path[-stealth]
		(m-2-3) edge node [] {} (m-1-3)
		(m-3-1) edge node [] {} (m-2-3)
		(m-3-2) edge node [] {} (m-2-3)
		(m-3-3) edge node [] {} (m-2-3)
		(m-3-4) edge node [] {} (m-2-3)
		(m-3-5) edge node [] {} (m-2-3)
		(m-3-6) edge node [] {} (m-2-3)
		(m-4-1) edge node [] {} (m-3-2)
		(m-4-2) edge node [] {} (m-3-2)
		(m-4-3) edge node [] {} (m-3-2)
		(m-4-4) edge node [] {} (m-3-3)
		(m-5-3) edge node [] {} (m-4-4)
		(m-5-4) edge node [] {} (m-4-4)
		(m-5-5) edge node [] {} (m-4-4);
\end{tikzpicture}
\end{center}

\section{Dynamics of the Fibonacci trace map}\label{sec:dynamics}

In the proof of Theorem \ref{thm:main} we shall employ dynamical properties of the Fibonacci trace map $f$, which we discuss in this section. For a brief overview of basic notions and notation from the theory of hyperbolic and partially hyperbolic dynamical systems, see Appendix \ref{b2}. We shall make more specific references to the appendix throughout the text, as necessary.

For convenience henceforth we shall refer to the sequence $\set{f^k({x})}_{k\in\N}$ as the \textit{positive}, or \textit{forward, semiorbit} of ${x}$, and denote it by $\mathcal{O}_f^+({x})$. The \textit{negative}, or \textit{backward, semiorbit} for $k\in\Z_{<0}$, and the \textit{full orbit}, for $k\in\Z$, are defined similarly and denoted, respectively, by $\mathcal{O}_f^-$ and $\mathcal{O}_f = \mathcal{O}_f^+\cup \mathcal{O}_f^-$.

Define the so-called \textit{Fricke-Vogt character} \cites{Fricke1896, Fricke1897, Vogt1889}
\begin{align}\label{part2_eq7}
I(x,y,z) = x^2 + y^2 + z^2 - 2xyz - 1.
\end{align}
Consider the family of cubic surfaces $\set{S_V}_{V\geq 0}$ given by
\begin{align}\label{part2_eq8}
S_V = \set{\vt{x}\in\R^3: I(\vt{x}) = V}.
\end{align}
For all $V > 0$, $S_V$ is a smooth, connected 2-dimensional submanifold of $\R^3$ without boundary. For $V = 0$, $S_V$ has four conic singularities,
\begin{align}\label{part2_eq9}
P_1 = (1,1,1),\text{\hspace{5mm}} P_2 = (-1,-1,1),\text{\hspace{5mm}} P_3 = (1,-1,-1),\text{\hspace{5mm}}P_4 = (-1,1,-1),
\end{align}
away from which the surface is smooth (see Figure \ref{part2_fig1}).
\begin{figure}[t]
\centering
 \subfigure[$V = 0.0001$]{
\includegraphics[scale=.3]{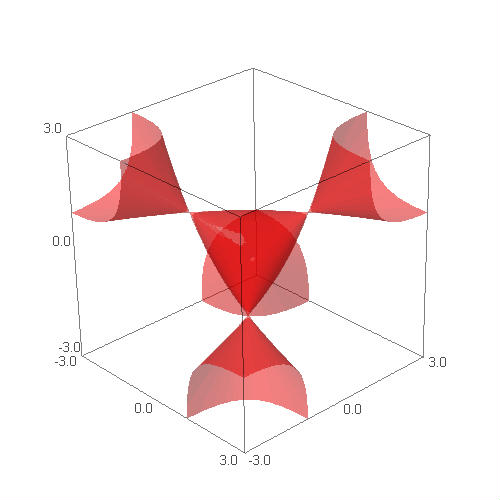}}
 \subfigure[$V = 0.01$]{
\includegraphics[scale=.3]{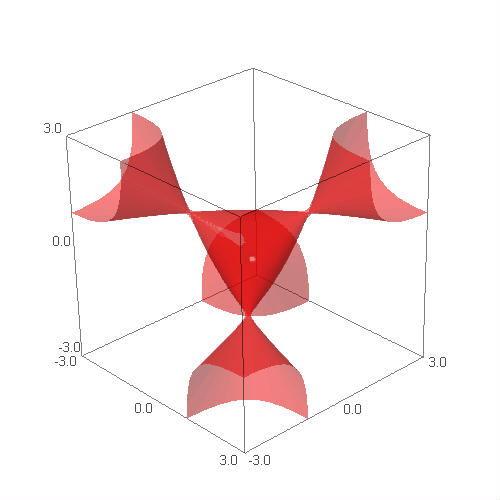}}
\\
 \subfigure[$V = 0.05$]{
\includegraphics[scale=.3]{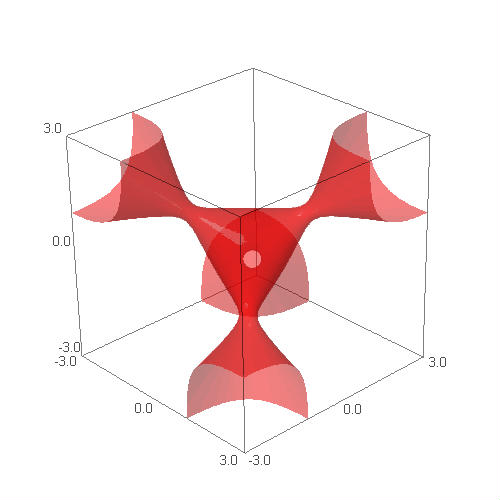}}
 \subfigure[$V = 1$]{
\includegraphics[scale=.3]{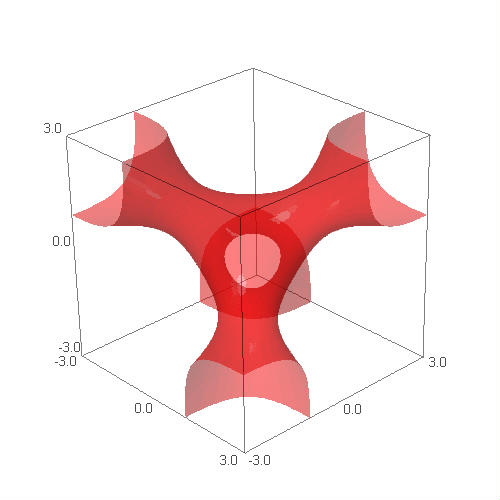}}
\caption{Invariant surfaces $S_V$ for four values of $V$.}
\label{part2_fig1}
\end{figure}

One can easily check that $f$ preserves the Fricke-Vogt character by verifying that $f(I(\vt{x})) = I(\vt{x})$; hence $f$ also preserves the surfaces $\set{S_V}$ (i.e. $f: S_V\rightarrow S_V$). For convenience we shall write $f_V$ for $f|_{S_V}$. In fact, since $f$ is invertible with the inverse $f^{-1}(x,y,z) = (y, z, 2yz - x)$, $f:S_V\rightarrow S_V$ is an analytic diffeomorphism. 

Since those points whose positive semiorbit is bounded play a crucial role in our analysis, it is convenient, for future reference, to state as a separate result the following necessary and sufficient conditions for a semiorbit to be bounded.

\begin{prop}\label{part2_prop0}
Let $x_k = \pi\circ f^k(x_1,x_0,x_{-1})$. We have the following.
\begin{enumerate}
\item Assume $\abs{x_{-1}}\leq C$ for some $C \geq 1$. The sequence $\set{x_k}_{k\geq -1}$ is unbounded if and only if there exists $k_0 \geq 0$ such that
\begin{align*}
\abs{x_{k_0 - 1}}\leq C\text{\hspace{5mm}and\hspace{5mm}}\abs{x_{k_0}},\abs{x_{k_0 + 1}} > C.
\end{align*}
\item A sufficient condition for $\set{x_k}_{k\geq -1}$ to be unbounded is that there exists $k_0\geq 0$ such that
\begin{align*}
\abs{x_{k_0}},\abs{x_{k_0 + 1}}> 1\text{\hspace{5mm}and\hspace{5mm}}\abs{x_{k_0}}\abs{x_{k_0 + 1}} > \abs{x_{k_0 - 1}}.
\end{align*}
\end{enumerate}
\end{prop}
\begin{rem} As defined earlier, $\pi$ denotes projection onto the third coordinate.
\end{rem}

\begin{proof}
For the proof of (1), see, for example, \cite[Proposition 5.2]{Damanik2005} (replace $1$ with $C$). For (2), see \cites{Kohmoto1983, Kadanoff1984} ((2) also follows from the aforementioned proof of (1)).
\end{proof}

\begin{rem}
For detailed analysis of orbits of trace maps, see \cite{Roberts1996}.
\end{rem}

Another result that we conveniently state as a separate statement establishes that the set of all points with bounded positive semiorbit is a closed set. This is a direct consequence of Proposition \ref{part2_prop0}.

\begin{lem}\label{part2_lem0_1}
For $V_0 \geq 0$ and $\infty \geq V_1\geq V_0$, the set of all points of $\bigcup_{V\in[V_0, V_1]}S_V$ whose forward semiorbit is bounded is a closed set.
\end{lem}
\begin{proof}
If $\mathcal{O}^+_f(x_1, x_0, x_{-1})$ is unbounded, then there exists a $k_0 \geq 0$ such that the point $f^{k_0}(x_1,x_0,x_{-1})$ satisfies (1) of Proposition \ref{part2_prop0}, and hence also satisfies (2), which is an open condition.
\end{proof}

Another consequence of Proposition \ref{part2_prop0} that will be used later is Proposition \ref{part2_prop2} below; for a proof see \cite[Proposition 5.2]{Damanik2005}, or (in a more general context) \cite{Roberts1996}.

\begin{prop}\label{part2_prop2}
The positive semiorbit is unbounded if and only if $\set{f^k({x})}_{k\in\N}$ escapes to infinity in every coordinate.
\end{prop}

Most conclusions about the dynamics of $f$ that we shall derive and use come from the knowledge of dynamics of $f$ on the surfaces $S_V$, i.e. dynamics of $f_V$ (not surprisingly, since these surfaces are invariant under $f$). In the following sections we recall some known results about dynamics of $f_V$, $V \geq 0$, as well as prove some new results necessary for the present investigation.

\subsection{Dynamics of \texorpdfstring{$f_V$} {fV} for \texorpdfstring{$V\geq 0$}{V>=0}}\label{subsec:dynamics-v}

In this and the following sections, we shall use the notation and terminology from Appendix \ref{b}.

\subsubsection{Hyperbolicity of \texorpdfstring{$f_V$} {fV} for \texorpdfstring{$V > 0$}{V > 0}}\label{subsec:hyperbolicity}

The following result on hyperbolicity of the Fibonacci trace map will serve as the main tool for us. For definition of a locally maximal transitive hyperbolic set (and the notion of $(1,1)$ splitting) see Section \ref{b1}.

\begin{thm}[M. Casdagli in \cite{Casdagli1986}, D. Damanik and A. Gorodetski in \cite{Damanik2009}, and S. Cantat in \cite{Cantat2009}\footnote{The special case of $V \geq 64$ was done by M. Casdagli in \cite{Casdagli1986}. D. Damanik and A. Gorodetski extended the result to all $V > 0$ sufficiently small in \cite{Damanik2009}. Finally, S. Cantat proved the result for all $V > 0$ in
\cite{Cantat2009} (D. Damanik and A. Gorodetski, and S. Cantat obtained their results independently, and used different techniques).}]\label{part2_thm1}
For $V > 0$ let
\begin{align*}
\Omega_V = \set{p\in S_V: \mathcal{O}_f(p)\text{ is bounded}}.
\end{align*}
Then $\Omega_V$ is a Cantor set, it is $f_V$-invariant compact locally maximal transitive hyperbolic set in $S_V$ (with $(1,1)$ splitting). Moreover, $\Omega_V$ is precisely the set of nonwandering points of $f_V$ (a point $p$ is nonwandering if for any neighborhood $U$ of $p$ and $N\in\N$, there exists $n \geq N$ such that $f^n(U)\cap U\neq \emptyset$).
\end{thm}

\begin{rem}\label{part2_rem5}
It follows that for $V > 0$, for any $x\in\Omega_V$, $W_\loc^\mathrm{s}(x)\cap\Omega_V$ and $W_\loc^\mathrm{u}(x)\cap\Omega_V$ are Cantor sets ($W_\loc^\mathrm{s}(x)$ denotes the local stable manifold at $x$, and $W_\loc^\mathrm{u}(x)$ denotes the local unstable manifold at $x$ -- see Appendix \ref{b1-2} for definitions and properties of these objects).
\end{rem}

\begin{rem} 
 In fact, $f_V$ satisfied what is called \textit{Smale's Axiom A} \cite{Smale1967}. The general theory of Axiom A diffeomorphisms is extensive and forms one of the central themes in the modern theory of dynamical systems. Suffice it to say that Axiom A diffeomorphisms are the hallmark of chaotic dynamics.
\end{rem}

The preceding theorem is of importance, because as a consequence of it, the set of points whose forward semiorbit is bounded can be endowed with a sensible geometric structure, which is the subject of Corollary \ref{part2_cor1} below. The corollary follows from general principles in hyperbolic dynamics (and has been employed implicitly in a number of previous works, especially \cite{Casdagli1986, Damanik2009}; for the technical details, see Corollary 2.5 and the discussion preceding it in \cite{Damanik0000my1}. In the statement of the corollary, $W^s(\Omega_V)$ stands for $\bigcup_{x\in\Omega_V} W^s(x)$, where $W^s(x)$ is the global stable manifold (in contrast to the local one from Remark \ref{part2_rem5} above). For details, see equations \eqref{part2_eq4} and \eqref{part2_eq5} in Section \ref{b1-2}.

\begin{cor}\label{part2_cor1}
For $V > 0$, $x\in S_V$, $\mathcal{O}_f^+(x)$ is bounded if and only if $x\in W^\mathrm{s}(\Omega_V)$.
\end{cor}

\subsubsection{Dynamics of \texorpdfstring{$f_V$}{fV} for \texorpdfstring{$V = 0$}{V = 0}}\label{part2_2_2}
As mentioned above, the surface $S_0$ is smooth everywhere except for the four singularities $P_1,\dots, P_4$ (see \eqref{part2_eq9} and Figure \ref{part2_fig1}). Let us set, and henceforth fix, the following notation
\begin{align}\label{center_part}
\mathbb{S} = \set{(x,y,z)\in S_0: \abs{x},\abs{y},\abs{z} \leq 1}.
\end{align}
It is easily seen that $\mathbb{S}$ is invariant under $f$. Moreover, $f|_{\mathbb{S}}$ is a factor of the hyperbolic diffeomorphism on $\mathbb{T}^2$,
\begin{align}\label{torus_map}
\mathcal{A}(\theta, \phi) = ((\theta + \phi),\theta)(\mathrm{mod\hspace{1mm}} 1),
\end{align}
via the factor map
\begin{align}\label{semiconjugacy}
F:(\theta, \phi)\mapsto (\cos 2\pi(\theta + \phi), \cos 2\pi\theta, \cos 2\pi\phi).
\end{align}
By \textit{a factor} we mean
\begin{align*}
f|_{\mathbb{S}}\circ F = F\circ\mathcal{A}.
\end{align*}
The map $F$ is not, however, a conjugacy in the sense of \eqref{part2_eq3} from Section \ref{b1-2}, since $F$ is a two-to-one map. The dynamics of $f$ on $\set{P_1,\dots,P_4}$ is as follows.
\begin{align*}
f: P_1\mapsto P_1; \text{\hspace{5mm}}f: P_2\mapsto P_3\mapsto P_4\mapsto P_2.
\end{align*}

It is easy to see that $\mathcal{A}:\mathbb{T}^2\rightarrow\mathbb{T}^2$ is a hyperbolic map with the full space $\mathbb{T}^2$ constituting a compact hyperbolic set (such systems are more commonly known as \textit{Anosov diffeomorphisms}). Even though $F:\mathcal{A}\rightarrow \mathbb{S}$ is not a conjugacy, the behavior of $Df$ on the tangent bundle of $\mathbb{S}$ away from the singularities $P_1,\dots, P_4$ still inherits the hyperbolic behavior from $D\mathcal{A}$ via $DF$ (this is stated more precisely in Lemma \ref{cone_projection} below). However, the singularities (or, better to say, dynamics near the singularities) requires special treatment. We now concentrate on a neighborhood of $P_1$; similar results hold for the other singularities due to the symmetries of $f$ (see Section \ref{a2} for details).

Let $\mathrm{Per}_2(f)$ denote the set of period-two periodic points for $f$ (note that $\mathrm{Per}_2(f)$ is precisely what is called $\rho_1$ in Section \ref{a1}). A direct computation shows that
\begin{align*}
\mathrm{Per}_2(f) = \set{(x,y,z): x\in\left(-\infty,\hspace{1mm}\frac{1}{2}\right)\cup\left(\frac{1}{2},\hspace{1mm}\infty\right),\hspace{1mm}y = \frac{x}{2x-1},\hspace{1mm}z = x}.
\end{align*}
Let
\begin{align}\label{eq_per_pnts}
\vartheta(x) = \left(x, \frac{x}{2x-1}, x\right), \hspace{2mm} \vartheta: \left(-\infty,\hspace{1mm}\frac{1}{2}\right)\cup\left(\frac{1}{2},\hspace{1mm}\infty\right)\rightarrow\R^3
\end{align}
be the curve of these periodic points (see Figure \ref{part3_fig1}). We have
\begin{figure}[t]
\centering
 \subfigure{
\includegraphics[scale=.3]{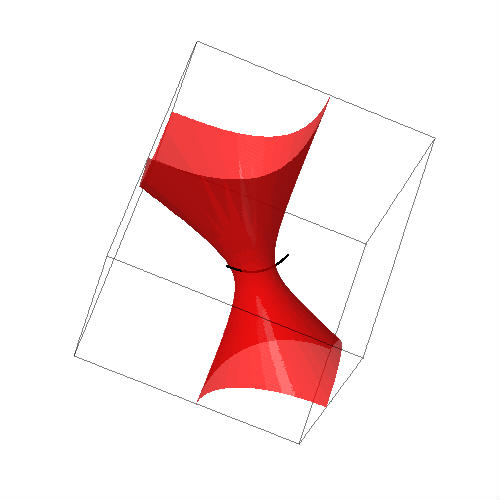}}
 \subfigure{
\includegraphics[scale=.3]{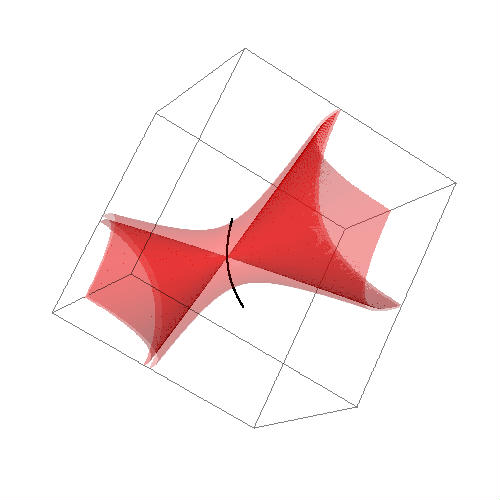}}
\caption{$\mathrm{Per}_2(f)$ in a neighborhood of $P_1$.}\label{part3_fig1}
\end{figure}
Also,
\begin{align}\label{part3_2_2_eq1}
I(\vartheta(x)) \geq 0,
\end{align}
with $I(\vartheta(x)) = 0$ if and only if $x = 1$, where $\vartheta(1) = P_1$. On the other hand,
\begin{align*}
Df_{P_1} = 
\begin{pmatrix}
2 & 2 & -1\\
1 & 0 & 0\\
0 & 1 & 0
\end{pmatrix}
\text{\hspace{2mm} which is similar to \hspace{2mm}}
\begin{pmatrix}
\frac{3 - \sqrt{5}}{2} & 0 & 0\\
0 & -1 & 0\\
0 & 0 & \frac{3 + \sqrt{5}}{2}
\end{pmatrix}.
\end{align*}
Hence $T_{P_1}\R^3$ splits as
\begin{align}\label{part3_2_2_eq2}
T_{P_1}\R^3 = E_{P_1}^\mathrm{s}\oplus E_{P_1}^c\oplus E_{P_1}^\mathrm{u},
\end{align}
where $E^s_{P_1}$ corresponds to the subspace spanned by the eigenvector corresponding to the largest eigenvalue of $Df$ (which is strictly larger than one), $E^u_{P_1}$ corresponds to the subspace spanned by the eigenvector of $Df_{P_1}$ corresponding to the smallest eigenvalue of $Df$ (which is the reciprocal of the largest one), and $E^c_{P_1}$ corresponds to the subspace spanned by the eigenvector of $Df_{P_1}$ corresponding to the eigenvalue $1$, which is also the tangent space of $\mathrm{Per}_2(f)$ at the point $P_1$ (see Section \ref{a1} for a slightly more detailed discussion). Now using invariance of $\mathrm{Per}_2(f)$ under $f$, and combining \eqref{part3_2_2_eq1} with Theorem \ref{part2_thm1} for $x\neq 1$ and \eqref{part3_2_2_eq2} for $x = 1$, we get that the curve $\vartheta$ is normally hyperbolic in a neighborhood of $P_1$ (for definitions and properties, see Section \ref{b3}). Hence we can apply Theorem \ref{a1_thm1} and, in the notation of Section \ref{b3} (see equation \eqref{eq:cscu-def})
, we get that $W_\loc^\mathrm{cs}(\mathrm{Per}_2(f))$ and $W_\loc^\mathrm{cu}(\mathrm{Per}_2(f))$ are smooth two-dimensional submanifolds of $\R^3$ (roughly speaking, the manifold $W_\loc^{\mathrm{cs}}(\mathrm{Per}_2(f))$ is precisely the set of those points of $\R^3$ which lie in some $\epsilon$-neighborhood of $\mathrm{Per}_2(f)$ and whose positive semiorbit does not leave this neighborhood, and similarly for $W_\loc^{\mathrm{cu}}(\mathrm{Per}_2(f))$ with $f$ replaced by $f^{-1}$; more details are given in Section \ref{b3}). Moreover, by invariance of the surfaces $S_V$ under $f$, it follows that $W_\loc^\mathrm{cs,cu}(\mathrm{Per}_2(f))\cap S_V$ is precisely the one-dimensional manifold $W_\loc^\mathrm{s,u}(\mathrm{Per}_2(f)\cap S_V)$ on $S_V$ (see Sections \ref{b1-2} and \ref{b3}), for $V > 0$. When $V=0$, $W_\loc^\mathrm{cs,cu}(\mathrm{Per}_2(f))\cap S_0\setminus\set{P_1}$ forms the local strong-stable/unstable one-dimensional submanifold (as defined in equation \eqref{eq:ss-manifold} of Section \ref{b3}
) on $S_0$ consisting of two smooth branches that connect smoothly (when viewed as submanifolds of $\R^3$) at $P_1$. Points of the strong-stable manifold converge to $P_1$ under iterations of $f$; the same happens under iterations of $f^{-1}$ for points on the strong-unstable manifold. The strong-stable and strong-unstable manifolds intersect at $P_1$.

We have
\begin{align*}
\frac{d(I\circ\vartheta)}{dx} \neq 0
\end{align*}
for all $x\neq 1$. Hence  $\vartheta$,  and therefore also $W_\loc^\mathrm{cs,cu}(\mathrm{Per}_2(f))$, intersect $S_V$ transversely for $V > 0$. On the other hand, $W_\loc^\mathrm{cs,cu}(\mathrm{Per}_2(f))$ has quadratic tangency with $S_0$ along the strong-stable/unstable submanifold (see also \cite[Section 4]{Damanik2009}). This is a crucial point to make, since the proof of Theorem \ref{thm:main} relies heavily on transversality of intersection of the line of initial conditions, $\gamma_{(J_0, J_1)}$, from \eqref{model_eq5}, with $W^\mathrm{cs}(\mathrm{Per}_2(f))$ (which is defined below). In order to prove transversality, we proceed by perturbation analysis starting with $\gamma_{(J_0, J_1)}$ with $J_0 = J_1$. In this case $\gamma$ lies on $S_0$, and since $W^\mathrm{cs}(\mathrm{Per}_2(f))$ is tangent to $S_0$, $\gamma$ is tangent to $W^\mathrm{cs}(\mathrm{Per}_2(f))$ (this shows in particular that proving transversality for $J_0 \neq J_1$ is not a trivial task). We then show that if $J_0/J_1$ is close 
to $1$, and not equal to $1$, $\gamma_{(J_0, J_1)}$ is transversal to $W^\mathrm{cs}(\mathrm{Per}_2(f))$. The situation is further complicated by the fact that we do not have any explicit, or quantitative, knowledge about $W^{\mathrm{cs}}(\mathrm{Per}_2(f))$, and must rely only on qualitative analysis.

The manifolds $W_\loc^\mathrm{cs}(\mathrm{Per}_2(f))$ and $W_\loc^\mathrm{cu}(\mathrm{Per}(f))$ can be extended globally to
\begin{align*}
W^{cs}(\mathrm{Per}_2(f)) = \bigcup_{n\in\N}f^{-n}(W_\loc^\mathrm{cs}(\mathrm{Per}_2(f)))\end{align*}
and
\begin{align*}
W^\mathrm{cu}(\mathrm{Per}_2(f)) = \bigcup_{n\in\N}f^{n}(W_\loc^\mathrm{cu}(\mathrm{Per}_2(f))).
\end{align*}
In this case $W^\mathrm{s,u}(\mathrm{Per}_2(f)\cap S_V) = W^\mathrm{cs,cu}(\mathrm{Per}_2(f))\cap S_V$ for $V > 0$. For $V=0$, these form branches of global strong-stable and strong-unstable submanifolds of $S_0$ that connect at $P_1$ - these branches are injectively immersed one-dimensional submanifolds of $S_0$.

Similar results hold for $P_2, P_3, P_4$. Indeed, as $V$ takes on positive values, the points $P_2, P_3$ and $P_4$ bifurcate from three cycles to six cycles. These six cycles form three smooth curves, one through each $P_2, P_3$ and $P_4$. Considering $f^6$ in place of $f$, it is easy to show, as in the case of $P_1$ above, that each curve is normally hyperbolic. For further details, refer to Section \ref{a1}. 

Let us fix the following notation. Denote by $\mathbb{W}_i^\mathrm{s,u}$ the center-stable/unstable 2-dimensional invariant manifold to the normally hyperbolic curve through $P_i$. We denote by $\mathbb{W}_{i,\loc}^\mathrm{s,u}$ a small neighborhood of the normally hyperbolic curve in $\mathbb{W}_i^\mathrm{s,u}$. As an aside, we should mention that in fact the normally hyperbolic curve is precisely $\mathbb{W}_{i, \loc}^s\cap \mathbb{W}_{i,\loc}^u$, provided that the neighborhood $\mathbb{W}_{i,\loc}^{s,u}$ is taken sufficiently small, and this intersection is transversal (we should note that there exist, in fact \textit{infinitely many}, other intersections between $\mathbb{W}_i^s$ and $\mathbb{W}_i^u$).

In particular, the orbit $\mathcal{O}^{+}_{f_0}({x})$ (respectively, $\mathcal{O}^{-}_{f_0}({x})$), for ${x}\in S_0$, is bounded if and only if either ${x}\in\mathbb{S}$ or ${x}\in \mathbb{W}^\mathrm{s}_i\cap S_0$ (respectively, ${x}\in\mathbb{W}^\mathrm{u}_i\cap S_0$) for some $i \in\set{1,\dots, 4}$.

\subsection{Partial hyperbolicity: center-stable and center-unstable manifolds}\label{part2_2_3}

In the previous section we proved normal hyperbolicity of $f$ on certain submanifolds (namely, the curves of periodic points through the singularities) and derived the existence of two-dimensional analogs of stable and unstable manifolds. The result of the previous section is a special case of a more general fact, which is the subject of Proposition \ref{part2_prop7}. For the definition of a \textit{partially hyperbolic set} (as well as the notion of $(1,1,1)$ splitting), see Section \ref{b2}. Note also that the results of the previous section cannot be encapsulated into the following proposition, because in what follows, the results are stated for surfaces $S_V$, $V > 0$; due to presence of singularities $P_1,\dots, P_4$ in $S_0$, we had to treat the case of $S_0$ separately.  
\begin{prop}\label{part2_prop7}
Let
\begin{align*}
\mathcal{M} = \bigcup_{V\in(0, \infty)}S_V.
\end{align*}
For $V_1 \geq V_0 > 0$ let
\begin{align*}
\Omega = \bigcup_{V\in[V_0, V_1]}\Omega_V.
\end{align*}
Then $\mathcal{M}$ is a smooth, connected, $f$-invariant 3-dimensional submanifold of $\R^3$, $\Omega\subset\mathcal{M}$ is compact, $f$-invariant partially hyperbolic set with $(1,1,1)$ splitting, and the following holds.

There exist two families, denoted by $\mathcal{W}^\mathrm{s}$ and $\mathcal{W}^\mathrm{u}$, of smooth 2-dimensional connected manifolds injectively immersed in $\mathcal{M}$, whose members we denote by, respectively, $W^\mathrm{cs}$ and $W^\mathrm{cu}$, and call \textit{center-stable} and \textit{center-unstable} manifolds, with the following properties.
\begin{enumerate}
\item The family $\mathcal{W}^\mathrm{s,u}$ is $f$-invariant;
\item For every $x\in \Omega$ there exist unique $W^\mathrm{cs}(x)\in\mathcal{W}^{s}$ and $W^\mathrm{cu}(x)\in\mathcal{W}^\mathrm{u}$ such that $x\in W^\mathrm{cs}(x)\cap W^\mathrm{cu}(x)$;
\item Conversely, for every $W^\mathrm{cs}\in\mathcal{W}^\mathrm{s}$ and $W^\mathrm{cu}\in\mathcal{W}^\mathrm{u}$, there exist $x, y\in\Omega$ such that $x\in W^\mathrm{cs}$ and  $y\in W^\mathrm{cu}$. In fact, if $W_1^\mathrm{cs}, W_2^\mathrm{cs}\in\mathcal{W}^\mathrm{s}$ and $W_1^\mathrm{cu}, W_2^\mathrm{cu}\in\mathcal{W}^\mathrm{u}$ with $W_1^\mathrm{cs}\cap W_2^\mathrm{cs}\neq\emptyset$ and $W_1^\mathrm{cu}\cap W_2^\mathrm{cu}\neq\emptyset$, then $W_1^\mathrm{cs} = W_2^\mathrm{cs} = W^\mathrm{cs}(x)$ for some $x\in\Omega$, and $W_1^\mathrm{cu} = W_2^\mathrm{cu} = W^\mathrm{cu}(y)$ for some $y\in\Omega$.
\item For any $V > 0$ and any $W^\mathrm{cs, cu}\in\mathcal{W}^\mathrm{s,u}$, $W^\mathrm{cs, cu}\cap S_V = W^\mathrm{s,u}(x)$ for every $x\in\Omega_V\cap W^\mathrm{cs, cu}$; moreover, this intersection is transversal.
\end{enumerate}
\end{prop}

\begin{proof}
All statements about $\mathcal{M}$, as well as compactness of $\Omega$, are trivially true. The surfaces $\set{S_V}$, $V > 0$, are all diffeomorphic. Fix $\widetilde{V} > 0$, and let $\pi_V: S_{\widetilde{V}}\rightarrow S_V$ be a diffeomorphism, depending smoothly on $V$, with $\pi_{\widetilde{V}} = \mathrm{Id}|_{S_{\widetilde{V}}}$. Now, after a smooth change of coordinates, $f:\mathcal{M}\rightarrow\mathcal{M}$ may be considered as a skew product of identity on an interval $I$ with a map on $S_{\widetilde{V}}$:

\begin{center}
\begin{tikzpicture}
 \matrix (m) [matrix of math nodes,
	      row sep = 3em,
	      column sep = 4em,
	      minimum width = 2em,text height=1.5ex, text depth=0.25ex]
{
  I\times S_{\widetilde{V}}	&	I\times S_{\widetilde{V}}\\
  \bigcup_V S_V			&	\bigcup_V S_V\\
};
\path[-stealth]
  (m-1-1)	edge node [left]	{$\widetilde{\pi}$}	(m-2-1)
		edge node [above]	{$\mathcal{G}$}		(m-1-2)
  (m-2-1)	edge node [above]	{$f$}			(m-2-2)
  (m-1-2)	edge node [right]	{$\widetilde{\pi}$}	(m-2-2);
\end{tikzpicture}
\end{center}

where $\widetilde{\pi}(V, x) = \pi_V(x)$, and $\mathcal{G}(V, x) = (V, \tilde{\pi}^{-1}\circ f \circ \tilde{\pi}(V, x))$. Now all statements about $\Omega$ follow from Theorem \ref{part2_thm1} (in particular, noting that $\Omega_{\widetilde{V}}$ is hyperbolic).

We now construct the family $\mathcal{W}^\mathrm{s}$. The family $\mathcal{W}^\mathrm{u}$ can be constructed similarly by considering $f^{-1}$ in place of $f$.

Fix $\widetilde{V} > 0$ and $x\in \Omega_{\widetilde{V}}$. Take $\delta > 0$ small (in particular so that $\widetilde{V} - \delta > 0$) and for $V\in [\widetilde{V}-\delta, \widetilde{V}+\delta]$ let $H_V: \Omega_{\widetilde{V}}\rightarrow\Omega_{V}$ be the topological conjugacy (see Section \ref{b1-1}). Then $W^\mathrm{cs}_\loc(x) = \bigcup_{V\in(\widetilde{V} - \delta, \widetilde{V}+\delta)}W_\loc^\mathrm{s}(H_V(x))$ is a smooth two-dimensional manifold (see \cite[Section 6]{Hirsch1977} for proof of smoothness) that can be extended to all $V$, and the sought $W^\mathrm{cs}(x)$ is then given by 
\begin{align*}
W^\mathrm{cs}(x) = \bigcup_{n\in\N} f^{-n}(W_\loc^\mathrm{cs}(x)).
\end{align*}
The collection $\set{W^\mathrm{cs}(x)}_{x\in\Omega_{\widetilde{V}}}$ gives the family $\mathcal{W}^\mathrm{s}$.

Fix $x\in\Omega$ and consider the curve $V\mapsto H_V(x)$ in a neighborhood of $x$. This curve is the intersection of $W^\mathrm{cs}_\loc(x)$ with $W^\mathrm{cu}_\loc(x)$, hence is smooth. Since $f_V$ depends smoothly, and hence Lipschitz-continuously (when restricted to a compact subset), on $V$, by \cite[Theorem 7.3]{Hirsch1970}, $H_V$ is the fixed point of a contracting map on a certain Banach space that depends Lipschitz-continuously on $V$. So by \cite[Theorem 1.1]{Hirsch1970}, if $V_0 > 0$ is fixed, there exists $C > 0$ such that for all $V$ close to $V_0$,
\begin{align*}
\frac{\norm{H_V(x) - H_{V_0}(x)}}{\abs{V - V_0}} \leq C,
\end{align*}
proving transversality of intersection of $\set{H_V(x)}_{V}$ with $S_{V_0}$. Hence $W^\mathrm{cs, cu}(x)$ intersects $S_{V}$, $V > 0$, transversely, as claimed.
\end{proof}
\begin{rem}\label{part2_rem9}
From the proof of Proposition \ref{part2_prop7} and Section \ref{b1-2} it is evident that the local center-stable and center-unstable manifolds at $x$, $W^\mathrm{cs}_\loc(x)$ and $W^\mathrm{cu}_\loc(x)$, depend continuously on the point $x$ in the $C^1$ topology. Consequently, by compactness we obtain uniform (in $x$) transversality of $W^\mathrm{cs, cu}_\loc(x)$ with the surface $S_V$, $V> 0$.
\end{rem}
As a consequence of Proposition \ref{part2_prop7}, we obtain
\begin{prop}\label{part2_cor9}
Let $\mathcal{M}$ be as in Proposition \ref{part2_prop7}. Suppose $\gamma: [0, 1]\rightarrow\mathcal{M}$ is smooth and regular. Suppose further that $\gamma$ intersects the members of $\mathcal{W}^\mathrm{s}$ transversely with its endpoints not lying on center-stable manifolds. Then the set of all points $x\in[0,1]$ for which $\mathcal{O}^+_f(f\circ\gamma(x))$ is bounded is a Cantor set.
\end{prop}
\begin{proof}
By construction of the center-stable manifolds, for $s\in[0,1]$, $\gamma(s)$ has a bounded forward orbit if and only if $\gamma(s)$ belongs to a center-stable manifold.

Compactness follows from Lemma \ref{part2_lem0_1}. Now absence of isolated points and total disconnectedness follow from Remark \ref{part2_rem5} and Proposition \ref{part2_prop7}.
\end{proof}

We are ready to prove Theorem \ref{thm:main}.

\section{Proof of main results}\label{part3_1}

In this section we prove Theorem \ref{thm:main}. The proof of the theorem is inherently technical and involves nontrivial notions from the theory of hyperbolic and partially hyperbolic dynamical systems. To make the reading easier where possible, we have included as a separate section the necessary notions from the theory of dynamical systems in Appendix \ref{b}. Throughout this section, references are made to the appendix where appropriate. The reader is also advised to use Section \ref{sec:strategy} as a road map.

\subsection{A short comparison and contrast with Schr\"odinger operators}\label{sec:cc}

Those readers who are familiar with the results and methods of spectral theory of discrete one-dimensional quasi-periodic Schr\"odinger operators may justly inquire as to why we couldn't simply apply the methods that have already been developed for Schr\"odinger operators (see \cite{Damanik200X} for a broad overview). To answer this question, let us quickly recall the setup for the Schr\"odinger operator, denoted by $H$. The operator acts on $\ell^2(\Z)$ as
\begin{align*}
 (H\phi)_n = \phi_{n-1} + \phi_{n+1} + V\omega_n\phi_n,
\end{align*}
where $V \in\R$, and $\omega$ is the sequence obtained by performing the Fibonacci substitution on the letters $0$ and $1$, starting with, say, $0$: $0\mapsto 01\mapsto 010\mapsto 01001\mapsto\cdots$. This sequence can then be naturally extended to the left. For details, see the recent survey \cite{Damanik200X} and references therein. It turns out that there exists $\gamma: \R\rightarrow\R^2$, explicitly given by
\begin{align*}
 \gamma(E) = \left(\frac{E - V}{2}, \frac{E}{2}, 1\right),
\end{align*}
such that $E$ belongs to the spectrum of $H$ if and only if $\mathcal{O}^+_f(\gamma(E))$ is bounded. On the other hand, we have $I(\gamma(E)) = V^2/4$, which is clearly independent of $E$ and is non-negative (and is zero if and only if $V = 0$, which corresponds to the \textit{free Laplacian} case, whose spectrum is $[-2, 2]$). Thus the action of $f$ needs to be considered only on one surfaces, $S_{V^2/4}$, for any chosen and fixed $V > 0$ (this is not to say that the problem is trivial -- far from it!). In our case, however, as we shall soon see, the value of the invariant $I$ depends on the spectral parameter, forcing us to consider the action of $f$ on all $S_{V \geq 0}$ at once. It turns out that this is also what is responsible for multifractality (i.e. nonuniform local scaling properties) of the spectrum (in contrast to the case of Schr\"odinger operators). Let us conclude by mentioning that we have encountered the same difficulties (with the same consequences of multifractality) in a few other models (
see \cite{Damanik0000my1, Yessen2011a, Yessen2012a}).

\subsection{Preliminary technical platform}

The appearance of $\lambda^2$ in $\gamma_{(J_0,J_1)}$ in \eqref{model_eq5} makes $\gamma$ symmetric in $\lambda$ with respect to the origin. By abuse of notation, let us write $\lambda$ in place of $\lambda^2$, where $\lambda$ is allowed to take values in $[0,\infty)$.

Take $r = J_0/J_1$ and, for convenience, let us also write $\gamma_r$ in place of $\gamma_{(J_0,J_1)}$. Hence

\begin{align}\label{eq:new_gamma}
\gamma_r(\lambda) = \left(\frac{\lambda - (1 + J_1^2)}{2J_1},\frac{\lambda - (1 + r^2J_1^2)}{2rJ_1},\frac{1+r^2}{2r}\right).
\end{align}

\begin{prop}\label{part2_prop11}
For every $J_1 > 0$, there exists $r_0 = r_0(J_1)\in(0,1)$, such that for all $r\in(1 - r_0, 1 + r_0)$ and $r\neq 1$, the curve $\gamma_r$ in \eqref{eq:new_gamma} intersects the center-stable manifolds transversely.
\end{prop}
\begin{proof}
We begin by showing that $\gamma_r$ intersects uniformly transversely the manifolds $\mathbb{W}^\mathrm{s}_i$, $i \in \set{1,\dots,4}$. We shall see later (actually, we won't prove this but point to the work of S. Cantat \cite{Cantat2009}) that $\mathbb{W}_i^s$, for $i\in\set{1,\dots,4}$, form a dense (in the $C^1$ topology) subfamily of the family $\mathcal{W}^s$ from Proposition \ref{part2_prop7}, and the conclusion of Proposition \ref{part2_prop11} will follow. We shall then combine this result with Proposition \ref{part2_cor9}. 

\begin{lem}\label{part2_lem12}
For all $r$ sufficiently close to one and not equal to one, $\gamma_r$ intersects $\mathbb{W}^\mathrm{s}_i$, for $i\in\set{1,\dots,4}$, uniformly transversely.
\end{lem}
\begin{proof}
Returning to the map $\mathcal{A}$ in \eqref{torus_map}, we see that $\mathcal{A}$ is hyperbolic and is given by the matrix $A = \left(\begin{smallmatrix}1 & 1\\ 1 & 0\end{smallmatrix}\right)$ with eigenvalues
\begin{align*}
\mu = \frac{1 + \sqrt{5}}{2}\text{\hspace{5mm}and\hspace{5mm}}-\mu^{-1}=\frac{1 - \sqrt{5}}{2}.
\end{align*}
Let us denote by $\vt{v}^\mathrm{s}, \vt{v}^\mathrm{u}\in\R^2$ the stable and unstable eigenvectors of $A$:
\begin{align*}
A\vt{v}^\mathrm{s} = -\mu^{-1}\vt{v}^\mathrm{s},\hspace{2mm}A\vt{v}^\mathrm{u} = \mu\vt{v}^\mathrm{u},\hspace{2mm}\norm{\vt{v}^\mathrm{s,u}} = 1.
\end{align*}
Fix some small $\zeta > 0$ (in general we want $\zeta < 1$) and define the stable and unstable cone fields on $\R^2$ in the following way:
\begin{align}\label{cones_on_torus}
K^\mathrm{s}(p)& = \set{\vt{v}\in T_p\R^2: \vt{v} = v^\mathrm{u}\vt{v}^\mathrm{u} + v^\mathrm{s}\vt{v}^\mathrm{s}, \abs{v^\mathrm{s}}\geq\zeta^{-1}\abs{v^\mathrm{u}}},\\
K^\mathrm{u}(p)& = \set{\vt{v}\in T_p\R^2: \vt{v} = v^\mathrm{u}\vt{v}^\mathrm{u} + v^\mathrm{s}\vt{v}^\mathrm{s}, \abs{v^\mathrm{u}}\geq\zeta^{-1}\abs{v^\mathrm{s}}}.\notag
\end{align}
These cone fields are invariant:
\begin{align*}
A(K_p^\mathrm{u})\subset\Int(K^\mathrm{u}(Ap))\text{\hspace{5mm}and\hspace{5mm}}A^{-1}(K^\mathrm{s}(p))\subset\Int(K^\mathrm{s}(A^{-1}p)).
\end{align*}
Also, the iterates of the map $A$ expand vectors from the unstable cones, and the iterates of the map $A^{-1}$ expand vectors from the stable cones. That is, there exists a constant $C > 0$ such that
\begin{align*}
&\forall\vt{v}\in K^\mathrm{u}(p)\text{\hspace{5mm}}\forall n\in\N\text{\hspace{5mm}} \norm{A^n\vt{v}} > C\mu^n\norm{\vt{v}},\\
&\forall\vt{v}\in K^\mathrm{s}(p)\text{\hspace{5mm}}\forall n\in\N\text{\hspace{5mm}} \norm{A^{-n}\vt{v}} > C\mu^n\norm{\vt{v}}.
\end{align*}
The families of cones $\set{K^\mathrm{s}(p)}_{p\in\R^2}$ and $\set{K^\mathrm{u}(p)}_{p\in\R^2}$ can also be considered on $\mathbb{T}^2 = \R^2/\Z^2$.

The differential of the semiconjugacy $F$ in \eqref{semiconjugacy}, $DF$, sends these cone families to $Df$-invariant stable and unstable cone families on $\mathbb{S}\setminus\set{P_1,\dots,P_4}$. Let us denote these images by $\set{\mathcal{K}^\mathrm{s}}$ and $\set{\mathcal{K}^\mathrm{u}}$, respectively. It is clear that away from a neighborhood of singularities, the cones $\mathcal{K}^{s,u}$ have nonzero size. The following lemma sates that the size of these cones is uniformly bounded away from zero on $\mathbb{S}\setminus\set{P_1,\dots, P_4}$.

\begin{lem}[{\cite[Lemma 3.1]{Damanik2009}}]\label{cone_projection}
The differential of the semiconjugacy $DF$ induces a map of the unit bundle of $\mathbb{T}^2$ to the unit bundle of $\mathbb{S}\setminus\set{P_1,\dots,P_4}$. The derivatives of the restrictions of this map to a fiber are uniformly bounded. In particular, the sizes of cones in families $\set{\mathcal{K}^\mathrm{s}}$ and $\set{\mathcal{K}^\mathrm{u}}$ are uniformly bounded away from zero.
\end{lem}
%
%
%
%
\begin{figure}[t]
\centering
\includegraphics[scale=.4]{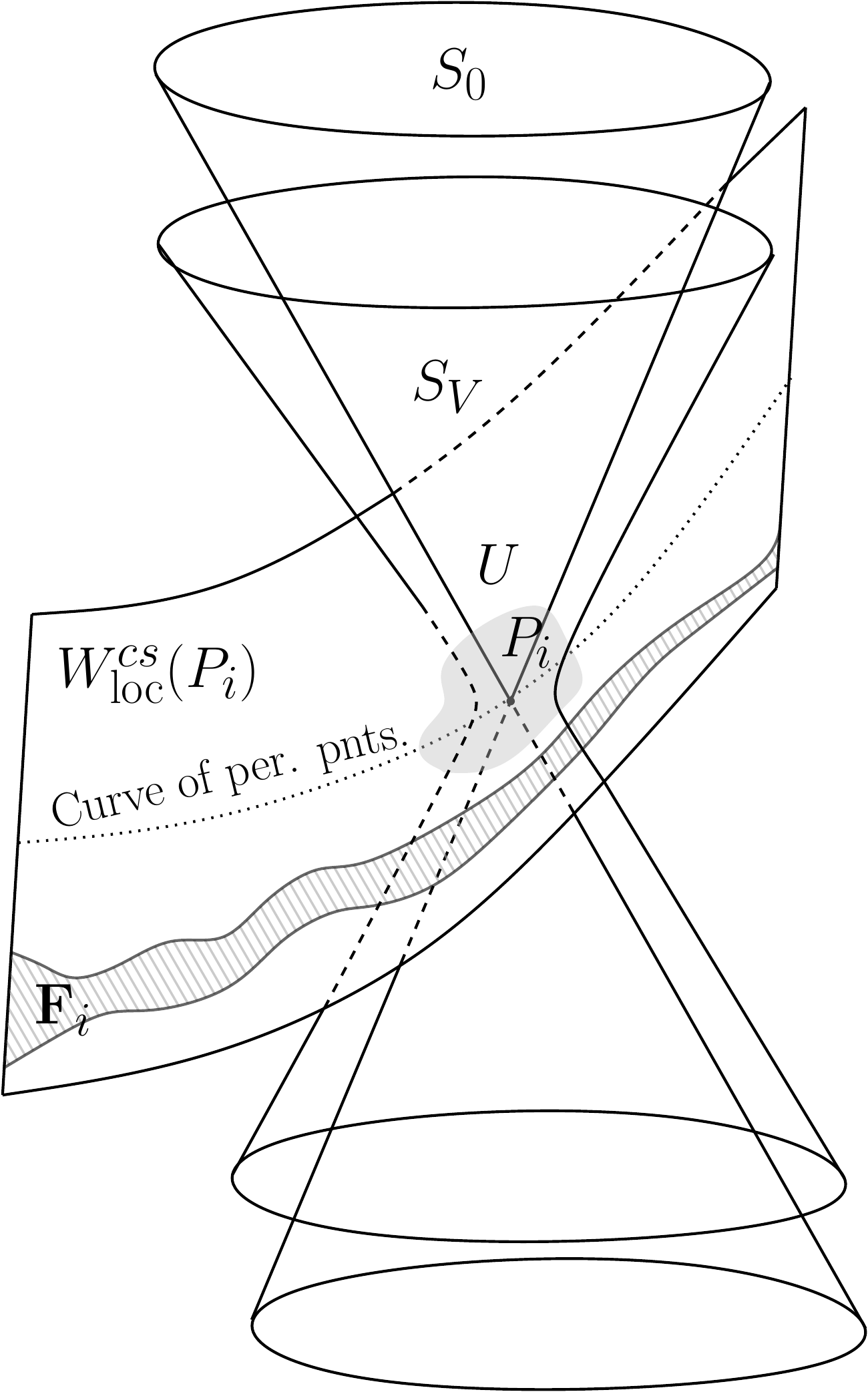}
\caption{}
\label{fig_fund_domain}
\end{figure}
Fix $j_0\in \set{1,\dots, 4}$. Let $\alpha$ be a smooth curve in $\mathbb{W}^s_{j_0}$ that is $C^1$ close to, and disjoint from, the curve of periodic points through the singularity $P_{j_0}$. Assume also that $\alpha\cap S_0\in\mathbb{S}$. Let $\mathbb{F}_{j_0}$ be the fundamental domain in $\mathbb{W}^s_{j_0}$ bounded by $\alpha$ and $f^{-1}(\alpha)$ (for the definition of a fundamental domain, see Section \ref{b1-3}). Let $U$ be a neighborhood of $\set{P_1,\dots,P_4}$ in $\R^3$ so small, that $U\cap \mathbb{F}_{j_0} = \emptyset$ (see Figure \ref{fig_fund_domain}). 

Given $V_0 > 0$ sufficiently small, for all $V\in [0, V_0)$, the surface $S_V\setminus U$ consists of five smooth connected components (with boundary), one of which is compact. Let $\mathbb{S}_{V, U}$ denote the compact component. The family $\set{\mathbb{S}_{V, U}}_{V\in [0, V_0)}$ depends smoothly on $V$, that is, there exists a family of smooth projections, depending smoothly on $V$: 
\begin{align*}
\set{\pi_V: \mathbb{S}_{0, U}\rightarrow \mathbb{S}_{V, U}}_{V\in [0, V_0)},\text{\hspace{5mm}}\pi_0 = \mathrm{Id}_{\mathbb{S}_{0,U}}.
\end{align*}
Assuming $V_0$ is sufficiently small, $D\pi_V$ carries the cones $\set{\mathcal{K}^\mathrm{s,u}}$ to nonzero cones on $\mathbb{S}_{V, U}$; denote these cones by $\set{\mathcal{K}^\mathrm{s,u}_V}$.

In the next series of lemmas we shall construct in some special way some families of cone fields and show that they are invariant under the action by $Df$. The reason for doing this is the following. As we have already mentioned, we do not have any quantitative information about $\mathbb{W}^s_i$. In particular, we cannot check directly whether the given line, $\gamma_r$, is transversal to $\mathbb{W}^s_i$. However, based on qualitative properties of $\mathbb{W}_i^s$ and quantitative information that \textit{we do have} about the surfaces $S_V$, $V\geq 0$, and the line $\gamma_r$, we can construct a cone field with the following properties. 

\begin{itemize}
 \item This cone field is transversal to $\mathbb{W}_i^s$ at points in the set $\mathbb{F}_i\cap S_V$ for all $V> 0$ sufficiently small (here we use quadratic tangency of $\mathbb{W}_i^s$ with $S_0$);
 
 \item The line $\gamma_r$ falls inside this cone field for all $r$ sufficiently close to $1$ (we can check this directly from the explicit expression of $\gamma$);
 
 \item This cone field is invariant under the action by $Df$ (we use both, the geometry of $S_V$ and some dynamical properties of $f$).
\end{itemize}

We then take any point in $\gamma_r\cap \mathbb{W}_i^s$, say $\gamma_r(x)$, and iterate it by $f$ until it falls inside $\mathbb{F}_i$; say $f^k(\gamma_r(x))\in \mathbb{F}_i$. By the invariance of the constructed cone field, we must then have $Df^k(\gamma'(x))$ inside the cone at $f^k(x)$, which is transversal to $\mathbb{W}_i^s$. Thus $\gamma'(x)$ must have been transversal to $T_{\gamma(x)}\mathbb{W}^s_i$ to begin with.

Below, Lemmas \ref{cones_on_compacta}, \ref{invariance_of_cones} and \ref{uniform_expansion} establish invariance of $\set{\mathcal{K}_V^{s,u}}$, as well as scaling properties of vectors from these cones under the action by $Df$, by considering different cases. The final result is recorded in Corollary \ref{invariance_and_expansion}.

\begin{lem}\label{cones_on_compacta}
For all $N\in\N$ there exists $V_0 > 0$ sufficiently small such that for all $k\in\Z_{\geq 0}$ with $k\leq N$ and all $V\in [0, V_0)$, if $x\in \mathbb{S}_{V, U}$ and $f^k(x)\in \mathbb{S}_{V, U}$, then
\begin{align*}
Df^k_x(\mathcal{K}^\mathrm{u}_V(x))\subset\Int(\mathcal{K}^\mathrm{u}_V(f^k(x))).
\end{align*}
\end{lem}
\begin{proof}
Since $\pi_V$ depends smoothly on $V$, for $x\in\mathbb{S}_{0,U}$, the cones $\mathcal{K}^\mathrm{u}(x)$ and $\mathcal{K}^\mathrm{u}_V(\pi_V(x))$ are close provided that $V$ is close to zero. Since $f_V$ depends smoothly on $V$, $Df_V^k$ and $Df_0^k$ are close. In particular, for a given $x\in\mathbb{S}_{0,U}$ with $f^k(x)\in\mathbb{S}_{0,U}$, if $\pi_V(x)\in\mathbb{S}_{V,U}$ and $f^k(\pi_V(x))\in\mathbb{S}_{V,U}$, then $Df^k_V(\mathcal{K}^\mathrm{u}_V(\pi_V(x)))$ and $Df^k(\mathcal{K}^\mathrm{u}(x))$ are close. Thus by compactness of the surfaces $\mathbb{S}_{V,U}$, we can choose $V_0$ suitably small, so that the conclusion of the Lemma holds.
\end{proof}
\begin{lem}[{\cite[Lemma 5.4]{Damanik2009}}]\label{invariance_of_cones}
Assuming $U$ is sufficiently small, there exists sufficiently large $N \in \N$ and sufficiently small $V_0 > 0$, such that for all $V\in [0, V_0)$ the following holds. If $x\in \mathbb{S}_{V, U}$ and $k$ is the smallest number such that $f^k(x)\in \mathbb{S}_{V,U}$ and $k\geq N$, then
\begin{align*}
Df^k_x(\mathcal{K}^\mathrm{u}_V(x))\subset\Int(\mathcal{K}^\mathrm{u}_V(f^k(x))).
\end{align*}
\end{lem}
\begin{lem}[{\cite[Lemma 5.2]{Damanik2009}}]\label{uniform_expansion}
There exists $V_0 > 0$ sufficiently small, $C > 0$ and $\mu > 1$, such that for all $V\in [0, V_0)$ the following holds. If $x\in \mathbb{S}_{V, U}$ and for $k\in\N$, $f^k(x)\in\mathbb{S}_{V, U}$, and if $\vt{v}\in \mathcal{K}_V^\mathrm{u}(x)$, then
\begin{align*}
\norm{Df^k_x(\vt{v})}\geq C\mu^k\norm{\vt{v}}.
\end{align*}
\end{lem}
Combination of Lemmas \ref{cones_on_compacta}, \ref{invariance_of_cones}, \ref{uniform_expansion} gives
\begin{cor}\label{invariance_and_expansion}
Assuming $U$ is small, there exists $V_0 > 0$ sufficiently small, $C > 0$ and $\mu > 1$ such that for all $V\in[0, V_0)$, the following holds. If $x\in \mathbb{S}_{V,U}$, $\vt{v}\in\mathcal{K}_V^\mathrm{u}(x)$, $k\in\N$ and $f^k(x)\in\mathbb{S}_{V,U}$, then
\begin{align*}
Df^k_x(\mathcal{K}^\mathrm{u}_V(x))\subset\Int(\mathcal{K}^\mathrm{u}_V(f^k(x)))\text{\hspace{5mm}and\hspace{5mm}}\norm{Df^k_x(\vt{v})}\geq C\mu^k\norm{\vt{v}}.
\end{align*}
\end{cor}
With $U$ and $V_0$ satisfying the hypothesis of Corollary \ref{invariance_and_expansion}, let us construct the following cone field on $\mathbb{S}_{V, U}$, for $V\in [0, V_0)$ and $\eta > 0$:
\begin{align}\label{3d_cones}
K_V^\eta(x) = \set{(\vt{u},\vt{v})\in T_x\mathbb{S}_{V, U}\oplus (T_x\mathbb{S}_{V,U})^\perp: \vt{u}\in\mathcal{K}_V^\mathrm{u}(x)\text{, }\norm{\vt{v}}\leq\eta\sqrt{V}\norm{\vt{u}}}.
\end{align}

The cone field from \eqref{3d_cones} is of central importance, as it is the one that will be shown to be transversal to $\mathbb{W}_i^s$, and to contain the line $\gamma_r$. However, as it is defined, it may not be invariant in the sense of the preceding corollary. The next lemma establishes that it is \textit{almost invariant} (which is enough for us). 

\begin{lem}\label{invariance_3d_cones}
For every $\widehat{\eta}>0$ there exists $\eta = \eta(\widehat{\eta}) > 0$, $\eta < \widehat{\eta}$, and $V_0 > 0$ sufficiently small, such that for any $V\in [0, V_0)$, any $x\in \mathbb{S}_{V,U}$, $k\in\Z_{\geq 0}$, if $f^k(x)\in\mathbb{S}_{V,U}$, then
\begin{align*}
Df^k_x(K_V^\eta(x))\subset K_V^{\widehat{\eta}}(f^k(x)).
\end{align*}
\end{lem}
\begin{proof}
Smooth dependence of the surfaces $\set{S_V}_{V > 0}$ on $V$ and invariance under $f$ implies the following.
\begin{lem}\label{bound_on_gradiants}
For any $V > 0$, $x\in S_V$ and $k\in\Z$, if $\vt{v}\in (T_x S_V)^\perp$, then
\begin{align}\label{scaling_gradiants}
\norm{\proj{(T_{f^k(x)}S_V)^\perp}{Df^k_x(\vt{v})}} = \frac{\norm{\nabla I(x)}}{\norm{\nabla I(f^k(x))}}\norm{\vt{v}},
\end{align}
where $\nabla I$ is the gradient of the Fricke-Vogt character (see \eqref{part2_eq7}). In particular, there exists $D > 0$ such that for all $V\in (0,V_0)$ and any $x\in \mathbb{S}_{V, U}$, if $f^k(x)\in\mathbb{S}_{V, U}$, then for every $\vt{v}\in (T_x S_V)^\perp$, we have
\begin{align*}
\norm{\proj{(T_{f^k(x)}S_{V,U})^\perp}{Df_x^k(\vt{v})}}\leq D\norm{\vt{v}}.
\end{align*}
In fact, we can take
\begin{align*}
D = \sup\set{\frac{\norm{\nabla I(x)}}{\norm{\nabla I(y)}}: x,y\in \mathbb{S}_{V,U}, V\in [0, V_0]} <\infty.
\end{align*}
\end{lem}
\begin{proof}
Let
\begin{align*}
\mathcal{M} = \bigcup_{V > 0}S_V.
\end{align*}
Integrate the gradient vector field $x\mapsto\nabla I(x)$ on $\mathcal{M}$, and let $\alpha_x$ denote a compact arc along the integral curve through $x$, say parameterized on $[-1,1]$ with $\alpha_x(0) = x$. Let $\beta = f^k(\alpha_x)$. Then
\begin{align*}
\norm{\nabla I(x)}^2 = (I\circ\alpha_x)'(0) = (I\circ\beta)'(0) = C\nabla I(f^k(x))\cdot\nabla I(f^k(x)),
\end{align*}
where $C\nabla I(f^k(x))$ is the projection of $\beta'(0)$ onto $(T_{f^k(x)}S_V)^\perp$, $C > 0$ a constant. Hence
\begin{align*}
\norm{C\nabla I(f^k(x))} = \frac{\norm{\nabla I(x)}}{\norm{\nabla I(f^k(x))}}\cdot\norm{\nabla I(x)}.
\end{align*}
\end{proof}

\begin{figure}[t]
\centering
\includegraphics[scale=.4]{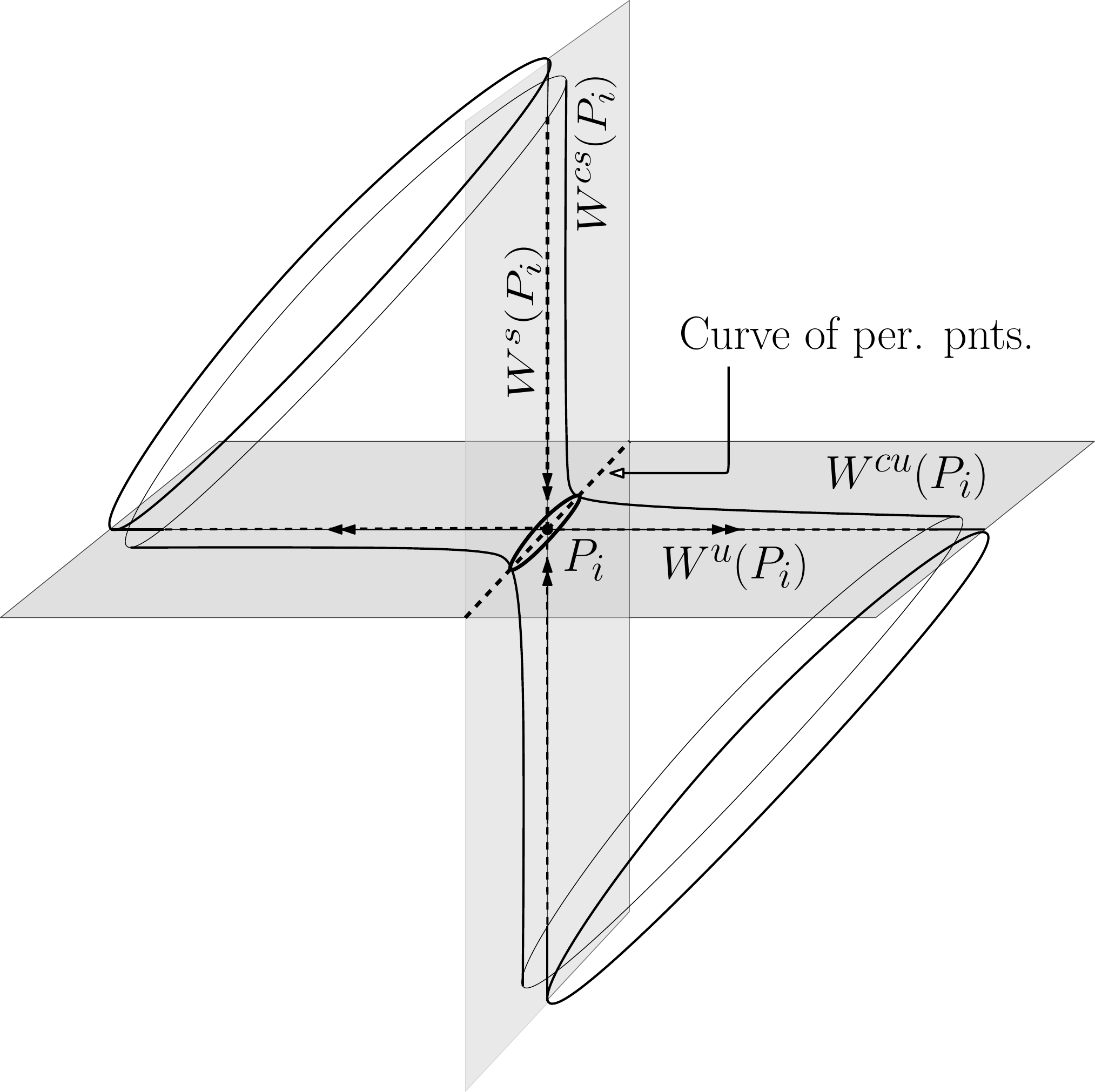}
\caption{}
\label{fig_rect}
\end{figure}

Let $D$ be as in Lemma \ref{bound_on_gradiants} and $\mu$ and $C$ as in Corollary \ref{invariance_and_expansion}. Let $k_0\in\N$ be the smallest number such that $C\mu^{k_0} > D$. Fix $N\in\N$ with $N > k_0$. Let $U^*$ be a neighborhood of $\set{P_1,\dots,P_4}$ in $\R^3$ such that $\overline{U^*}\subset U$, so small that if $x\in U^*$ and $l\in\N$ is the smallest number such that $f^l(x)\notin U$, then $l > N$.

\textit{Case (i)}. Suppose $x\in \mathbb{S}_{V, U}$, $f^k(x)\in\mathbb{S}_{V, U}$ and $N > k\geq k_0$. Then $C\mu^k > D$, hence the expansion in the cone $\mathcal{K}_V^\mathrm{u}(x)$ dominates the expansion along the normal. On the other hand, the normal at $x$, under the action of $Df^k_x$, may tilt to the side away from $\mathcal{K}_V^\mathrm{u}(f^k(x))$. However, since $k < N$, $\set{x,f(x),\dots,f^k(x)}\subset\mathbb{S}_{V,U^*}$. By compactness of $S_{V, U^*}$, the angle between the image under $Df^k_x$ of the normal at $x$ and $\mathbb{S}_{V,U^*}$ remains uniformly bounded away from zero. This, together with the fact that $\mathcal{K}_V^\mathrm{u}(x)$ is mapped into the interior of $\mathcal{K}_V^\mathrm{u}(f^k(x))$, allows us to choose $V_0$ sufficiently small to compensate for the tilt in the normal. Hence for any $\eta > 0$, there exists $V_0 > 0$ small, such that $Df^k_x(K_V^\eta(x))\subset K_V^\eta(f^k(x))$.

\textit{Case (ii)}. If $x\in \mathbb{S}_{V,U}$, $f^k(x)\in \mathbb{F}_{j_0}\subset \mathbb{S}_{V,U}$ and $k < N$, then for sufficiently small $\eta$ (depending only on $N$ and independent of $x$), $K_V^\eta(x)\subset K_V^{\widehat{\eta}}(f^k(x))$; that is, given that the number of iterations does not exceed a given constant, the distortion can be controlled. 

\textit{Case (iii)}. We now handle the case when $x$, under iterates of $f$, passes through $U^*$. By symmetries of the map $f$, it is enough to consider only a neighborhood of $P_1$.

Say $U^*$ is a neighborhood of $P_1$. If $U^*$ is sufficiently small, there exists a smooth change of coordinates $\Phi: U^*\rightarrow\R^3$ such that $\Phi(P_1) = (0,0,0)$ and the following holds.

Denote by $\mathbb{W}^\mathrm{s, u}_1(P_1)$ a small neighborhood of the point $P_1$ on the manifold $\mathbb{W}^\mathrm{s,u}_1$. We have
\begin{itemize}
\item $\Phi(\mathrm{Per}_2(f))$ is part of the line $\set{x = 0, z = 0}$;
\item $\Phi(\mathbb{W}^\mathrm{s}_1(P_1))$ is part of the plane $\set{z = 0}$;
\item $\Phi(\mathbb{W}^\mathrm{u}_1(P_1))$ is part of the plane $\set{x = 0}$;
\item $\Phi(\mathbb{W}^\mathrm{s}_1\cap S_0)$ is part of the line $\set{y = 0, z = 0}$;
\item $\Phi(\mathbb{W}^\mathrm{u}_1\cap S_0)$ is part of the line $\set{x = 0, y = 0}$.
\end{itemize}
Denote $\mathfrak{S}_V = \Phi(S_V)$. Then $\set{\mathfrak{S}_V}_{V>0}$ is a family of smooth surfaces depending smoothly on $V$, $\mathfrak{S}_0$ is diffeomorphic to a cone, contains lines $\set{y = 0, z = 0}$ and $\set{x = 0, y = 0}$, and at each nonzero point on those lines it has a quadratic tangency with the $xy$- and $zy$-plane (see Figure \ref{fig_rect}).

For a point $p$, denote its coordinates by $(x_p, y_p, z_p)$.
\begin{lem}[{\cite[Propositions 3.12 and 3.13]{Damanik2010}}]\label{technical}
Given $C_1, C_2 > 0$, $\rho > 1$, there exists $\delta_0 > 0$, $C > 0$, $\mu>1$ and $N_0\in\N$, such that for any $\delta\in(0,\delta_0)$, the following holds.

Let $g: \R^3\rightarrow\R^3$ be a $C^2$ diffeomorphism such that
\begin{enumerate}[(i)]
\item $\norm{g}_{C^2}\leq C_1$;
\item The plane $\set{z = 0}$ is invariant under the iterates of $g$;
\item $\norm{Dg(p) - A}<\delta$ for every $p\in\R^3$, where
\begin{align*}
A = \begin{pmatrix}
\rho^{-1} & 0 & 0\\
0 & 1 & 0\\
0 & 0 &\rho
\end{pmatrix}
\end{align*}
is a constant matrix. 
\end{enumerate}
Introduce the following cone field on $\R^3$:
\begin{align}\label{round_cones}
\widetilde{K}_p = \set{\vt{v}\in T_p\R^3, \vt{v} = \vt{v}_{xy}+\vt{v}_z:\abs{\vt{v}_z}\geq C_2\sqrt{\abs{z_p}}\norm{\vt{v}_{xy}}}.
\end{align}
Then for any $p\in\R^3$ satisfying $\abs{z_p} \leq 1$, 
\begin{enumerate}[(1)]
\item $Dg(\widetilde{K}_p)\subset\widetilde{K}_{g(p)}$;
\item if $\abs{z_{g^N(p)}}>1$ with $N\geq N_0$, then for any $\vt{v}\in\widetilde{K}_p$, if $Dg^N(\vt{v}) = \vt{u}_{xy}+\vt{u}_z$, then
\begin{align*}
\norm{\vt{u}_{xy}}<2\delta^{1/2}\abs{\vt{u}_z}\text{\hspace{5mm}and\hspace{5mm}}\norm{Dg^N(\vt{v})}\geq C\mu^N\norm{\vt{v}}.
\end{align*}
\end{enumerate}
\end{lem}
For a given $\delta > 0$, if the neighborhood $U^*$ of singularities is small enough, then at every point $p\in U^*$ the differential $D(\Phi\circ f\circ\Phi^{-1})(p)$ satisfies condition (iii) of Lemma \ref{technical}. Since the tangency of $\mathfrak{S}_0$ with the plane $\set{z = 0}$ is quadratic, there exists $C_2 > 0$ such that every vector tangent to $\mathfrak{S}_0$ from the cone $D\Phi(\mathcal{K}^\mathrm{u})$ also belongs to the cone in \eqref{round_cones}. The same holds for vectors tangent to $\mathfrak{S}_V$ from the cones $D\Phi(\mathcal{K}_V^\mathrm{u})$ for $V$ small enough. Therefore, Lemma \ref{technical} can be applied to all those vectors. In particular, we have
\begin{lem}\label{using_technical}
Assume $U$ is so small that Lemma \ref{technical} can be applied. There exists $C > 0$ such that if $x\in\mathbb{S}_{V, U}$, $f(x)\in U$ and $l\in\N$ is the smallest number such that $f^{l-1}(x)\in U$ and $f^l(x)\notin U$, then If $l > N_0$, then for any $\vt{v}\in T_x\R^3$ with $D\Phi(Df_x(\vt{v}))\in\widetilde{K}_{\Phi(f(x))}$, we have
\begin{enumerate}
\item $\proj{T_{f^l(x)}\mathbb{S}_{V,U}}{Df^l_x(\vt{v})}\in\mathcal{K}_V^\mathrm{u}(f^l(x))$;
\item $\norm{Df_x^l(\vt{v})}\geq C\mu^l\norm{\vt{v}}$, where $\mu$ is as in Lemma \ref{technical}.
\end{enumerate}
\end{lem}

Now, if $U$ is sufficiently small and $V_0 > 0$ is sufficiently small, and $x\in \mathbb{S}_{V, U}$ with $f(x)\in U$, then for any $\eta > 0$ we have
\begin{align*}
D\Phi(Df_x(K_V^\eta(x)))\subset \widetilde{K}_{\Phi(f(x))}.
\end{align*}
Hence Lemma \ref{using_technical} can be applied to vectors in $Df(K_V^\eta(x))$. In particular, choosing $U^*$ so small that $N > N_0$, we see that if $x\in \mathbb{S}_{V, U}$, $f(x)\in U$ and $k\in \N$ is the smallest number such that $f^k(x)\notin U$, and $\set{x, f(x),\dots,f^{k-1}(x)}\cap U^*\neq \emptyset$, then $Df^k_x(K_V^\eta(x))\subset K_V^\eta(f^k(x))$, for any $\eta > 0$.

The proof of Lemma \ref{invariance_3d_cones} follows by combining cases (i) and (ii), with
\begin{align*}
\eta = \widehat{\eta}\min\set{\frac{C}{D}, 1},
\end{align*}
where $C$ is as in Corollary \ref{invariance_and_expansion} and $D$ is as in Lemma \ref{bound_on_gradiants}. 
\end{proof}
We immediately obtain, as a consequence of Lemma \ref{invariance_3d_cones}, the following
\begin{cor}\label{cor_cone_invariance}
Since the fundamental domain $\mathbb{F}_{j_0}$ has quadratic tangency with $S_0$, there exists $V_0>0$ and $\widehat{\eta}>0$ such that for all $V\in[0,V_0)$ and $x\in\mathbb{S}_{V, U}\cap\mathbb{F}_{j_0}$, the cone $K_V^{\widehat{\eta}}(x)$ is transversal to $\mathbb{F}_{j_0}$. Let $\eta < \widehat{\eta}$ be as in Lemma \ref{invariance_3d_cones}. Then for all $x\in\mathbb{S}_{V,U}\cap\mathbb{W}_{j_0}^s$ and $\vt{v}\in K_V^\eta(x)$, if $x$ does not lie in the region bounded by $\alpha$ and the curve of periodic points (i.e. $x$ lies in $\bigcup_{n\geq 0}f^{-n}(\mathbb{F}_{j_0}))$, then $\vt{v}$ is transversal to $\mathbb{W}_{j_0}^s$.
\end{cor}
\begin{proof}
For every $x\in\mathbb{S}_{V,U}\cap\mathbb{W}_{j_0}^s$ that satisfies the hypothesis of the corollary, there exists $k\in\N$, such that $f^k(x)\in\mathbb{F}_{j_0}$. The result follows by Lemma \ref{invariance_3d_cones}.
\end{proof}
Now recall the definition of $\gamma_r(\lambda)$ from \eqref{eq:new_gamma}. With $I$ denoting the Fricke-Vogt character (see \eqref{part2_eq7}), we have
\begin{align}\label{invariant_at_gamma}
I(\gamma_r(\lambda))& = \frac{\lambda}{4}\left(\frac{1}{r} - r\right)^2;&\\
\frac{\partial I(\gamma_r(\lambda))}{\partial\lambda}& = \frac{1}{4}\left(\frac{1}{r} - r\right)^2\notag.&
\end{align}
Consequently $\gamma_1(\lambda)$ is the line that contains the singularities $P_1$ and $P_2$ (see \eqref{part2_eq9}) and lies entirely on $S_0$. For $r\neq 1$, $\gamma_r$ intersects the surfaces $\set{S_V}_{V\geq 0}$ transversely, intersecting each surface in a unique point, and intersects $S_0$ when $\lambda = 0$.

\textit{Case (i) (Assuming $J_1\neq 1$)}. Assume $J_1 \neq 1$. If $\lambda = 0$ then by (2) of Proposition \ref{part2_prop0},  $\mathcal{O}_f^+(\gamma_1(0))$ escapes. By continuity, there exists $r_0\in(0,1)$ such that for all $r\in(1-r_0,1+r_0)$, the set
\begin{align*}
\set{x\in\gamma_r: \mathcal{O}^+_f(x)\text{ is bounded }}
\end{align*}
lies on a line segment $\Lambda_r$ whose endpoints belong to $U$(the neighborhood of singularities) (see Figures \ref{fig_projection} and \ref{fig_gamma_1}). 
\begin{figure}[t]
\begin{minipage}[b]{0.4\linewidth}
\centering
\includegraphics[scale=.32]{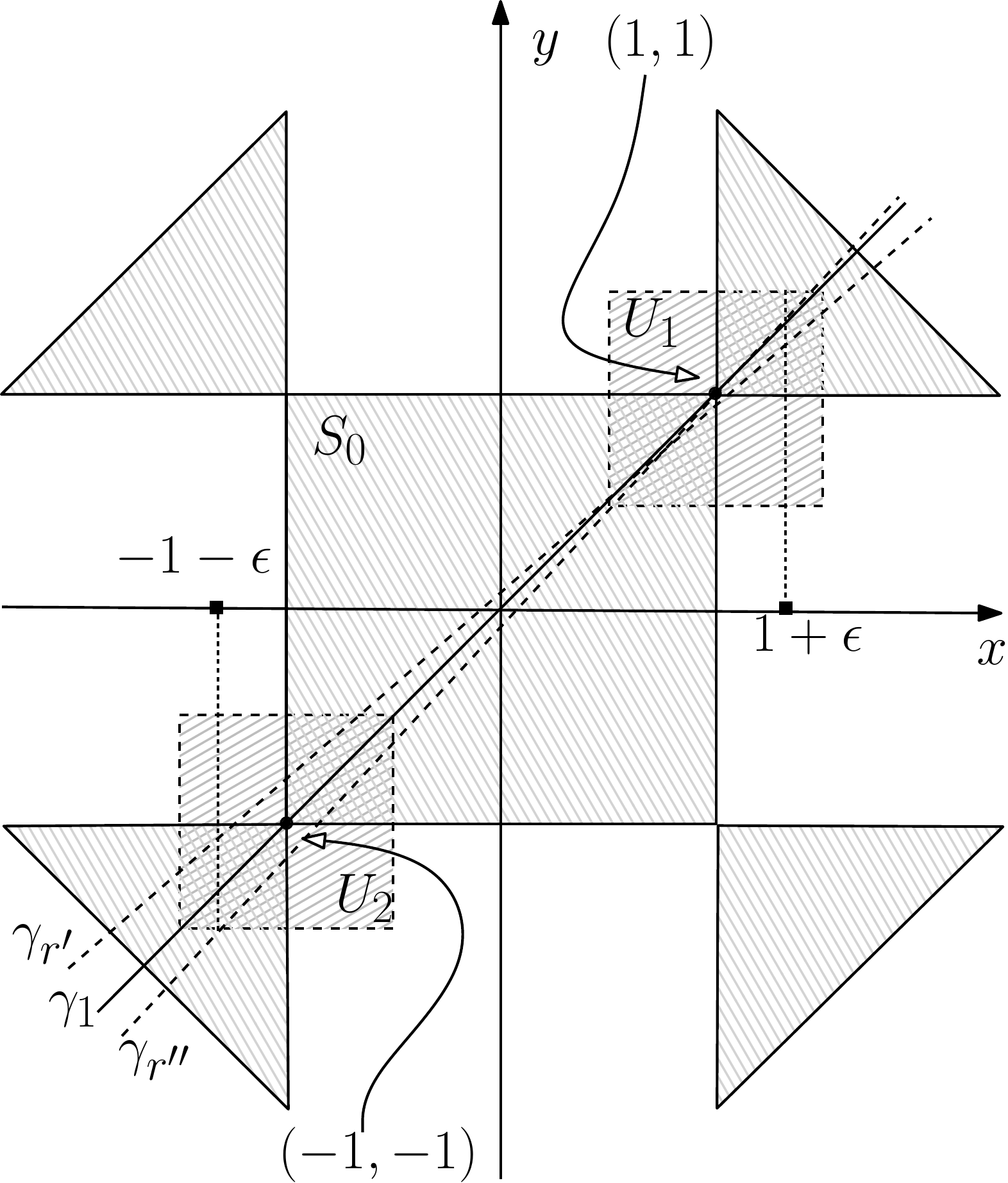}
\caption{
}
\label{fig_projection}
\end{minipage}
\hspace{0.5cm}
\begin{minipage}[b]{0.5\linewidth}
\centering
\includegraphics[scale=.32]{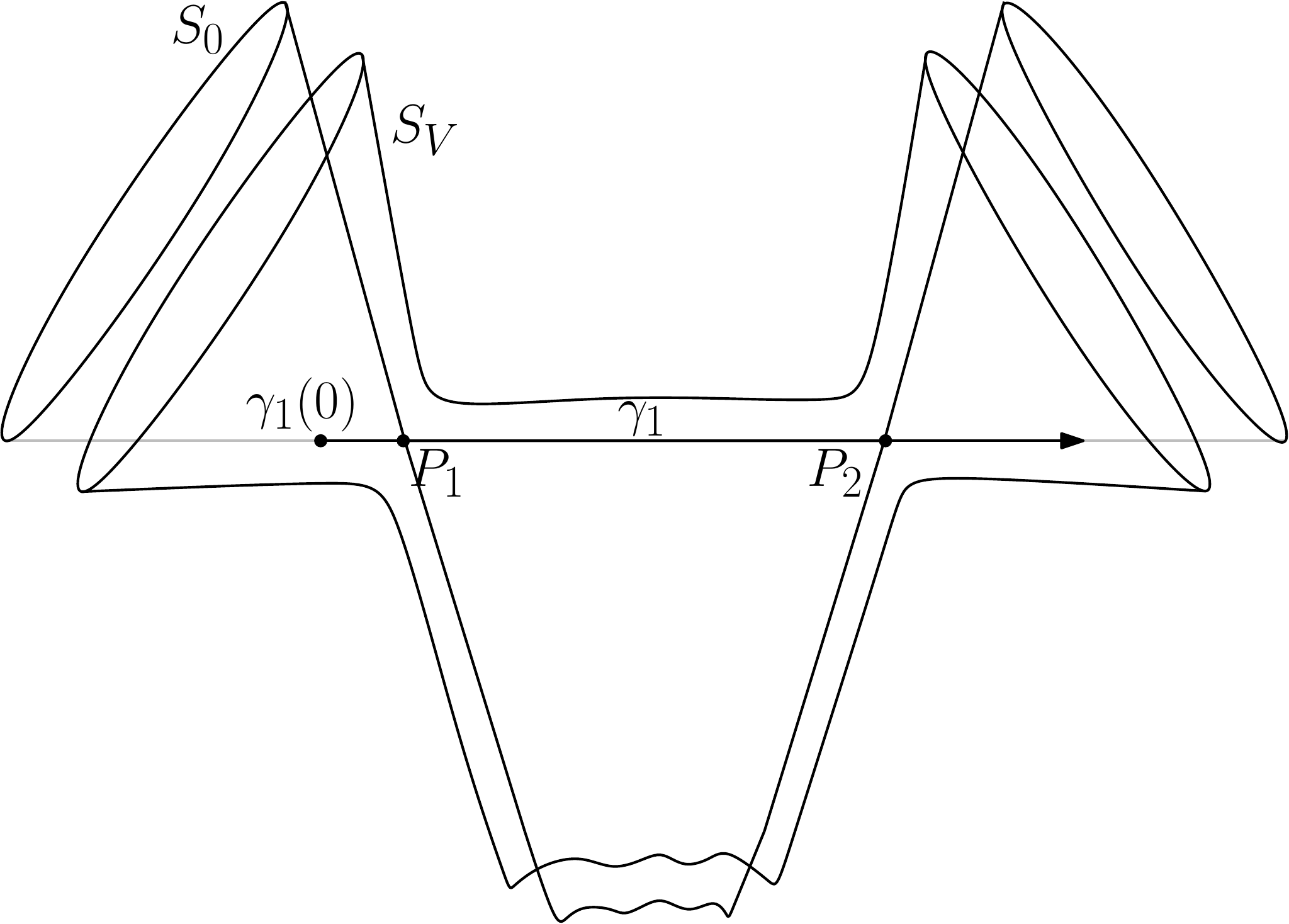}
\caption{
}
\label{fig_gamma_1}
\end{minipage}
\end{figure}
Let $V_0$ be so small that Corollary \ref{cor_cone_invariance} can be applied, and $r_0\in (0,1)$ such that for all $r\in(1- r_0, 1 + r_0)$, the line segment $\Lambda_r$ lies entirely in $U\cup\left(\bigcup_{V\in[0,V_0)}\mathbb{S}_{V,U}\right)$. Let $\eta$ and $\widehat{\eta}$ be as in Corollary \ref{cor_cone_invariance}. If $\lambda > 0$ is such that $\gamma_r(\lambda)\in S_V$, from \eqref{invariant_at_gamma} we have
\begin{align}\label{invariant_at_gamma2}
\frac{\partial I(\gamma_r(\lambda))}{\partial\lambda} = \frac{1}{\lambda}I(\gamma_r(\lambda)).
\end{align}
Since the set $\set{\lambda: \gamma_r(\lambda)\in\Lambda_r, r\in (1 - r_0, 1 + r_0)}$ is bounded away from zero uniformly in $r$, there exists $C_0 > 0$ such that for all such $\lambda$, if $\gamma_r(\lambda)\in\mathbb{S}_{V,U}$, then
\begin{align}\label{bound_on_angles}
\measuredangle(\gamma_r'(\lambda), S_V)\leq C_0V.
\end{align}
On the other hand, $F^{-1}(\gamma_1\cap\mathbb{S}) = \set{\theta = -\phi}$ (see \eqref{semiconjugacy}), and $(1, -1)$ is not an eigenvector of $\mathcal{A}$, hence by taking $K^\mathrm{u}(p)$ in \eqref{cones_on_torus} wider as necessary, we have that for all $\lambda$ such that $\gamma_r(\lambda)\in\mathbb{S}\setminus\set{P_1,\dots,P_4}$, $\gamma_r(\lambda)\in\Int(\mathcal{K}^\mathrm{u})$. By continuity we have, for all $r$ sufficiently close to $1$ and $V_0$ close to $0$, if $\lambda$ is such that $\gamma_r(\lambda)\in\mathbb{S}_{V,U}$, then
\begin{align}\label{projection_of_gamma_r}
\proj{T_{\gamma_r(\lambda)}S_V}{\gamma_r}\in\mathcal{K}_V^\mathrm{u}.
\end{align}
Combined with \eqref{bound_on_angles}, this gives $\gamma_r(\lambda)\in K_V^\eta(\gamma_r(\lambda))$.

Now suppose $\lambda$ is such that $\gamma_r(\lambda)\in\Lambda_r\cap U$. Let $k\in\N$ be the smallest number such that $f^k(\gamma_r(\lambda))\in\mathbb{S}_{V,U}$. If $k < N_0$, where $N_0$ is as in Lemma \ref{technical}, then $\gamma_r(\lambda)\notin U^*$. The cones $\mathcal{K}_V^\mathrm{u}$ defined on $\mathbb{S}_{V, U}$, $V\in [0, V_0)$, can be defined on $\mathbb{S}_{V, U^*}$ in the same way, by taking $V_0$ smaller as necessary. Hence by taking $r_0$ closer to $1$ as necessary, for all $r\in (1 - r_0, 1 + r_0)$, if $\lambda$ is such that $\gamma_r(\lambda)\in\mathbb{S}_{V, U^*}$, then \eqref{projection_of_gamma_r} again holds. Assuming $V_0$ was initially taken sufficiently small, we must have
\begin{align*}
Df^k_{\gamma_r(\lambda)}(\gamma'_r(\lambda))\in K_V^\eta(f^k(\gamma_r(\lambda)).
\end{align*}
Now suppose $k > N_0$.
\begin{lem}\label{lem_arb1}
For all sufficiently small $C^1$-perturbations of $\gamma_1$, $\Phi(\gamma_1)$ is tangent to the cones in \eqref{round_cones} on $\Phi(U)$, with $U$ a sufficiently small neighborhood of $P_1$.
\end{lem}
\begin{proof}
Observe that $\gamma_1$ lies in the plane $\set{z = 1}$. Notice, from \eqref{eq_per_pnts}, that the curve of periodic points, $\mathrm{Per}_2(f)$, is transversal to $\set{z = 1}$ at the point $P_1$. A simple calculation shows that the eigenvector corresponding to the smallest eigenvalue of $Df_{P_1}$ is also transversal to $\set{z = 1}$.  Hence $\Phi(\gamma_1)$ is transversal to $\mathbb{W}^u_1(P_1)$ (a neighborhood of $P_1$ in $\mathbb{W}^u_1$), and so also to the plane $\set{z = 0}$ in the rectified coordinates. Hence all sufficiently small $C^1$-perturbations of $\Phi(\gamma_1)$ are (uniformly) transversal to $\set{z = 0}$, and therefore tangent to the cones in \eqref{round_cones} in $\Phi(U)$, for sufficiently small $U$.
\end{proof}
Now Lemma \ref{using_technical} can be applied. We get
\begin{align}\label{projection_of_gamma_r2}
\proj{T_{f^k(\gamma_r(\lambda))}S_V}{Df^k_{\gamma_r(\lambda)}(\gamma'_r(\lambda))}\in \mathcal{K}_V^\mathrm{u}(f^k(\gamma_r(\lambda)).
\end{align}
On the other hand, by \eqref{bound_on_angles} we have
\begin{align*}
\norm{\proj{(T_{\gamma_r(\lambda)}S_V)^\perp}{\gamma'_r(\lambda)}}\leq \frac{C_0V}{\norm{\nabla I(\gamma_r(\lambda))}},
\end{align*}
and after applying \eqref{scaling_gradiants} (Lemma \ref{bound_on_gradiants}) we obtain
\begin{align}\label{gamma_normal}
&\norm{\proj{(T_{f^k(\gamma_r(\lambda))}S_V)^\perp}{Df^k_{\gamma_r(\lambda)}(\gamma'_r(\lambda))}}=\\
&= \frac{\norm{\nabla I(\gamma_r(\lambda))}}{\norm{\nabla I(f^k(\gamma_r(\lambda)))}}\norm{\proj{(T_{\gamma_r(\lambda)}S_V)^\perp}{\gamma'_r(\lambda)}}\leq \frac{C_0}{\norm{\nabla I(f^k(\gamma_r(\lambda)))}}V\notag.
\end{align}
By compactness, for all $V\in [0, V_0)$, $x\in \mathbb{S}_{V,U}$, $\norm{\nabla I(x)}$ is uniformly bounded away from zero. Thus, assuming $V_0$ is sufficiently small, combining equations \eqref{projection_of_gamma_r2} and \eqref{gamma_normal}, we get
\begin{align*}
Df^k_{\gamma_r(\lambda)}(\gamma'_r(\lambda))\in K_V^\eta(f^k(\gamma_r(\lambda))).
\end{align*}
We can now use Corollary \ref{cor_cone_invariance} to conclude that $\Lambda_r$ intersects $\mathbb{W}^\mathrm{s}_{j_0}$ transversely.

\textit{Case (ii) (Assuming $J_1 = 1$)}. When $J_1 = 1$, $\gamma_r(\lambda)$ has the form
\begin{align}\label{gamma_with_J_1}
\gamma_r(\lambda) = \left(\frac{\lambda - 2}{2}, \frac{\lambda - (1 + r^2)}{2r}, \frac{1 + r^2}{2r}\right).
\end{align}
Observe that if $r\neq 1$, then $\mathcal{O}^+_f(\gamma_r(0))$ escapes, since $f^2(\gamma_r(0))$ satisfies (2) of Proposition \ref{part2_prop0}. Hence for all $r\neq 1$ (and sufficiently close to $1$) there exists $\lambda_0(r) > 0$ such that for all $\lambda < \lambda_0(r)$, $\mathcal{O}^+_f(\gamma_r(\lambda))$ escapes, by Lemma \ref{part2_lem0_1}. We have
\begin{lem}\label{bound_on_v_to_lambda}
There exists $C_0 > 0$ such that for all $r\neq 1$ sufficiently close to $1$, and $\lambda_0(r)$ as above, if $\lambda \geq \lambda_0(r)$ and $\gamma_r(\lambda)\in S_V$, then
\begin{align*}
\frac{V}{\lambda} \leq C_0\sqrt{V}.
\end{align*}
\end{lem}
\begin{proof}
Let $\pi_1: \R^3\rightarrow \R$ denote projection onto the first coordinate. Since $r\approx 1$, $\lambda_0(r)\neq 0$ is close to zero, hence from \eqref{gamma_with_J_1},
\begin{align*}
\abs{\pi_1(\gamma_r(\lambda_0))} < 1.
\end{align*}
So by Proposition \ref{part2_prop0}, $\mathcal{O}^+_f(\gamma_r(\lambda_0))$ will diverge if $\abs{\pi_1\circ f^i(\gamma_r(\lambda_0))} > 1$, $i = 1, 2$. A simple calculation shows the existence of a constant $D > 0$, independent of $r$, such that
\begin{align*}
\abs{\pi_1\circ f(\gamma_r(\lambda_0))} \leq 1\text{\hspace{2mm}}or\text{\hspace{2mm}}\abs{\pi_1\circ f^2(\gamma_r(\lambda_0))} \leq 1\implies\lambda_0 \geq \frac{(r - 1)^2}{D}.
\end{align*}
Now, say $\gamma_r(\lambda)\in S_V$, and its forward orbit is bounded. Then $\lambda \geq \lambda_0(r)$. From \eqref{invariant_at_gamma},
\begin{align*}
V = \frac{\lambda}{4}\left(\frac{1}{r} - r\right)^2.
\end{align*}
So
\begin{align*}
\frac{\sqrt{V}}{\lambda} = \frac{\abs{r^2 - 1}}{2r\sqrt{\lambda}}\leq \frac{\abs{r^2 - 1}}{2r\sqrt{\lambda_0}}\leq \frac{\sqrt{D}(r + 1)}{2r}.
\end{align*}
The right side above is uniformly bounded for all $r$ away from zero. 
\end{proof}
Suppose $\gamma_r(\lambda)\in S_V$, $\gamma_r(\lambda)\in U^*\cap \Lambda_r$ and $k\in \N$ is the smallest number such that $f^k(\gamma_r(\lambda))\in\mathbb{S}_{V,U}$. Then $\lambda\neq 0$ and $k > N_0$. From equation \eqref{invariant_at_gamma2} and a calculation similar to \eqref{gamma_normal}, it follows that
\begin{align}\label{gamma_normal2}
&\norm{\proj{(T_{f^k(\gamma_r(\lambda))}S_V)^\perp}{Df^k_{\gamma_r(\lambda)}(\gamma'_r(\lambda))}}\notag\\
&=\norm{\proj{(T_{f^k(\gamma_r(\lambda))}S_V)^\perp}{Df^k_{\gamma_r(\lambda)}(\proj{(T_{\gamma_r(\lambda)}S_V)^\perp}{\gamma'_r(\lambda)}}}\\
&\leq \frac{D_0}{\lambda}V\leq D_0C_0\sqrt{V}\notag,
\end{align}
where $C_0$ is as in Lemma \ref{bound_on_v_to_lambda} and
\begin{align*}
D_0 = \sup\set{\frac{1}{\norm{\nabla I(x)}}: x\in\bigcup_{V\in[0,V_0)}\mathbb{S}_{V,U}}.
\end{align*}
On the other hand, $\norm{\gamma'_r(\lambda)} > 1$ independently of $r$. Now by Lemma \ref{lem_arb1}, Lemma \ref{using_technical} can be applied, so that by Part (2), with $C > 0$ and $\mu>1$ as in the Lemma, we obtain
\begin{align}\label{with_implication}
\norm{Df^k_{\gamma_r(\lambda)}(\gamma'_r(\lambda))}\geq C\mu^k\norm{\gamma'_r(\lambda)} > C\mu^k.
\end{align}
Hence we have
\begin{enumerate}
\item $\proj{T_{f^k(\gamma_r(\lambda))}S_V}{Df^k_{\gamma_r(\lambda)}(\gamma'_r(\lambda)}\in\mathcal{K}_V^\mathrm{u}(f^k(\gamma_r(\lambda)))$,
\item and from \eqref{gamma_normal2} and \eqref{with_implication},
\begin{align*}
\frac{\norm{\proj{(T_{f^k(\gamma_r(\lambda))}S_V)^\perp}{Df^k_{\gamma_r(\lambda)}(\gamma'_r(\lambda))}}}{\norm{Df^k_{\gamma_r(\lambda)}(\gamma'_r(\lambda))}}\leq \frac{D_0C_0}{C\mu^k}\sqrt{V}.
\end{align*}
\end{enumerate}
Hence
\begin{align*}
\measuredangle(Df^k_{\gamma_r(\lambda)}(\gamma'_r(\lambda)), S_V)\leq \widetilde{\eta}\sqrt{V},
\end{align*}
where $\widetilde{\eta}$ can be made arbitrarily small if $k$ is sufficiently large (i.e. if $U^*$ is initially chosen sufficiently small). Hence 
\begin{align*}
Df^k_{\gamma_r(\lambda)}(\gamma'_r(\lambda))\in K_V^\eta,
\end{align*}
and Corollary \ref{cor_cone_invariance} can be applied.

\textit{Case (iii) (If a point doesn't enter $\mathbb{F}_{j_0}$)}. Assume that the point $\gamma_r(\lambda)$ lies on the region of $\mathbb{W}^\mathrm{s}_{j_0}$ bounded by the curve of periodic points and the curve $\alpha$ (see discussion following Lemma \ref{cone_projection}). This is a finite region, hence contains at most finitely many points of intersection. So there exists $r_0\in(0,1)$, such that for all $r\in (1 - r_0, 1 + r_0)$, if $\gamma_r$ intersects this region, then this intersection is transversal. 

Combination of Cases (i), (ii) and (iii) gives transversality in Lemma \ref{part2_lem12}. To prove uniform transversality, observe that by Lemma \ref{invariance_3d_cones}, the cones $\set{K_V^\eta}$ are transversal to $\set{\mathbb{W}^\mathrm{s}_i}_{i\in\set{1,\dots,4}}$ for $V > 0$ sufficiently small. It follows that the cones $\set{K_V^{\eta/2}}$ are uniformly transversal to $\set{\mathbb{W}^\mathrm{s}_i}_{i\in\set{1,\dots,4}}$ on surfaces away from $S_0$. Hence by taking $r$ closer to one as necessary, but not equal to one, we can ensure that $\gamma_r$ lies in the cones $\set{K_V^{\eta/2}}$, and hence intersects $\set{\mathbb{W}^\mathrm{s}_i}_{i\in\set{1,\dots,4}}$ uniformly transversely. This finally concludes the proof of Lemma \ref{part2_lem12}.
\end{proof}
For justification of the following result, see \cite[Section 2]{Casdagli1986} and \cite[Section 5.4]{Cantat2009}.
\begin{lem}\label{part2_lem13}
For $V > 0$, $\set{\mathbb{W}^\mathrm{s}_i\cap S_V}_{i\in\set{1,\dots,4}}$ is dense (in the $C^1$ topology) in the lamination $W^\mathrm{s}(\Omega_V)$.
\end{lem}

Now from our construction of the family $\mathcal{W}^s$ in Proposition \ref{part2_prop7}, together with Remark \ref{part2_rem9}, combined with Lemma \ref{part2_lem13}, we get that $\mathbb{W}^s_{i\in\set{1,\dots,4}}$ is dense (in the $C^1$ topology) in the family $\mathcal{W}^s$ from Proposition \ref{part2_prop7}. Proposition \ref{part2_prop11} now follows from Lemma \ref{part2_lem12}.
\end{proof}

\subsection{Proof of Theorem \ref{thm:main}-(i)}\label{sec:proof-main-i}

The proof of Theorem \ref{thm:main}-(i) relies entirely on the dynamics of the trace map $f$, avoiding any spectral-theoretic considerations. We start with

\begin{prop}\label{part2_cor10}
With $\gamma_r$ from \eqref{eq:new_gamma}, define
\begin{align*}
\widetilde{B}_{\infty,r} = \set{\lambda\in [0,\infty): \mathcal{O}^+_f(\gamma_r(\lambda))\text{ is bounded}}.
\end{align*}
Define also
\begin{align*}
\widetilde{\sigma}_{k,r} = \set{\lambda\in[0,\infty): \abs{\pi\circ f^k(\gamma_r(\lambda))}\leq 1},
\end{align*}
where $\pi$, as above, denotes projection onto the third coordinate. Then $\widetilde{\sigma}_k\xrightarrow[k\rightarrow\infty]{}\widetilde{B}_\infty$ in Hausdorff metric.
\end{prop}

\begin{proof}
For convenience we drop $r$ and simply write $\widetilde{B}_\infty$ and $\widetilde{\sigma}_k$, keeping in mind implicit dependence of these sets on $r$.

We begin the proof with the following simple observation.
\begin{lem}\label{part2_cor10-helper1}
For any $V > 0$, if $x = (x_1,x_2,x_3)\in\Omega_V$, then one of the coordinates of $x$ is strictly smaller than one in absolute value.
\end{lem}

\begin{proof}
If all three coordinates are equal to one in absolute value, then $x$ has $-1$ as one of its coordinates, and $1$ for the other two (otherwise, $x$ is one of the four singularities of $S_0$, but by hypothesis $V > 0$). Then iterating forward twice gives a point that satisfies (2) of Proposition \ref{part2_prop0}, hence $x$ cannot be an element of $\Omega_V$.

If at least one but not all of its coordinates is equal to one in absolute value, and the rest are strictly greater than one in absolute value, then by applying either $f$ or $f^{-1} = (y,z,2yz - x)$, we shall obtain a point $(y_1,y_2,y_3)$ with $\abs{y_1},\abs{y_2}>1$ and $\abs{y_1 y_2}>\abs{y_3}$ or $\abs{y_2},\abs{y_3}>1$ and $\abs{y_2y_3}>\abs{y_1}$. In the first case, by (2) of Proposition \ref{part2_prop0}, the point $(y_1,y_2,y_3)$ has unbounded forward semiorbit; in the second case, by a similar result applied to $f^{-1}$, $(y_1,y_2,y_3)$ has unbounded backward semiorbit. In either case, $x$ cannot belong to $\Omega_V$. 
\end{proof}

By construction of the center-stable manifolds, it follows that $\widetilde{B}_\infty$ is precisely the intersection of $\gamma_r$ with the center-stable manifolds. As in the proof of Corollary \ref{cor_cone_invariance}, this intersection occurs on a compact line segment $\Lambda_r$ along $\gamma_r$. Now application of Lemma \ref{part2_lem0_1} shows that $\widetilde{B}_\infty$ is compact.

By Lemma \ref{part2_cor10-helper1} it follows that for all $\lambda\in \widetilde{B}_\infty$, there exists $N_\lambda\in\N$, such that for all $k\geq N_\lambda$, $f^k(\gamma(\lambda))$ belongs to $\widetilde{\sigma}_{k}\cup\widetilde{\sigma}_{k+1}\cup\widetilde{\sigma}_{k + 2}$. By compactness of $\widetilde{B}_\infty$, there exists such $N\in\N$ uniformly for all $\lambda\in \widetilde{B}_\infty$. Define
\begin{align*}
\Sigma_k = \widetilde{\sigma}_{N+k}\cup\widetilde{\sigma}_{N+k+1}\cup\widetilde{\sigma}_{N+k+2}.
\end{align*}
It is a simple observation that follows from Proposition \ref{part2_prop0} (see, for example, \cite{Damanik2005}) that $\Sigma_k\supset\Sigma_{k+1}$ for all $k\geq 0$, and $\widetilde{B}_\infty = \bigcap_{k\geq 0}\Sigma_k$. It follows that $\Sigma_k\xrightarrow[k\rightarrow\infty]{}\widetilde{B}_\infty$ in Hausdorff metric. It is therefore sufficient to prove that
\begin{align}\label{eq:part2_cor10}
\limsup_{k\rightarrow\infty}\hdist(\widetilde{\sigma}_{N+k},\Sigma_k) = 0.
\end{align}

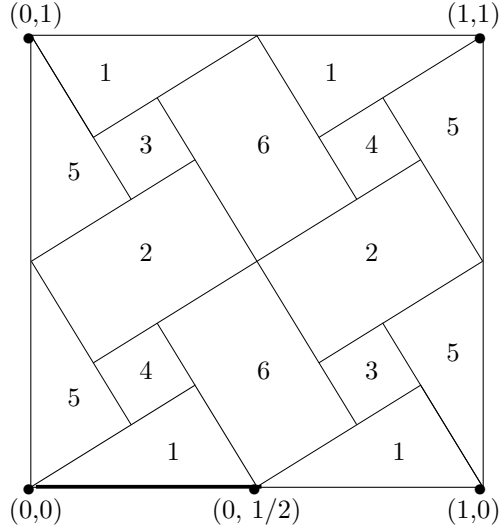
\begin{figure}
\centering
\setlength{\unitlength}{0.6mm}
\begin{picture}(120,120)

\put(10,10){\framebox(100,100)}

\put(5,3){(0,0)}

\put(50,3){(0, 1/2)}

\put(102,3){(1,0)}

\put(5,113){(0,1)}

\put(102,113){(1,1)}

\put(8,8){$\bullet$}

\put(108,8){$\bullet$}

\put(58,8){$\bullet$}

\put(8,108){$\bullet$}

\put(108,108){$\bullet$}

\put(10,10){\line(161,100){36.2}}

\put(60,10){\line(-61,100){13.8}}

\put(40,16){$1$}

\linethickness{0.5mm}

\put(11,10){\line(1,0){50}}

\linethickness{0.075mm}

\put(60,10){\line(161,100){36.2}}

\put(110,10){\line(-61,100){13.8}}

\put(90,16){$1$}

\put(60,110){\line(-161,-100){36.2}}

\put(10,110){\line(61,-100){13.8}}

\put(25,100){$1$}

\put(110,110){\line(-161,-100){36.2}}

\put(60,110){\line(61,-100){13.8}}

\put(75,100){$1$}

\put(10,60){\line(61,-100){22.1}}

\put(18,28){$5$}

\put(10,60){\line(161,100){22.1}}

\put(10,110){\line(61,-100){22.1}}

\put(18,78){$5$}

\put(110,60){\line(-61,100){22.1}}

\put(102,38){$5$}

\put(110,60){\line(-161,-100){22.1}}

\put(110,10){\line(-61,100){22.1}}

\put(102,88){$5$}

\put(60,60){\line(-161,-100){36.3}}

\put(60,60){\line(161,100){36.3}}

\put(60,60){\line(-61,100){22.1}}

\put(60,60){\line(61,-100){22.1}}

\put(88,46.35){\line(-161,-100){14.3}}

\put(32,73.65){\line(161,100){14.3}}

\put(37.88,46.28){\line(61,-100){8.4}}

\put(82.12,73.72){\line(-61,100){8.4}}

\put(34,34){$4$}

\put(84,84){$4$}

\put(34,84){$3$}

\put(84,34){$3$}

\put(34,60){$2$}

\put(84,60){$2$}

\put(60,34){$6$}

\put(60,84){$6$}

\end{picture}
\caption{The Markov partition for
$T|_{\mathbb{S}}$ (picture taken from \cite{Damanik2009}).}\label{fig:Casdagli-Markov}
\end{figure}
%
%
%

Recall that $f|_\mathbb{S}$ is a factor of the toral automorphism $\mathcal{A}: \mathbb{T}^2\rightarrow\mathbb{T}^2$ defined in \eqref{torus_map}, and the factor map $F$ is given in \eqref{semiconjugacy}. The Markov partition for $\mathcal{A}$ on $\mathbb{T}^2$ is given in Figure \ref{fig:Casdagli-Markov} (see \cite{Casdagli1986} and \cite{Damanik2009} for more details). This Markov partition is carried to a Markov partition on $\mathbb{S}$ for $f$ by the factor map $F$.

Let $\Lambda_1$ be the line segment along $\gamma_1$ that connects the singularities $P_1$ and $P_2$. Then $\Lambda_1$ is precisely the set of those points on $\gamma_1$ whose forward orbit under iterations of $f|_{S_0}$ is bounded. The set $F^{-1}(\Lambda_1)$ is the line segment connecting $(0,0)$ and $(0,1/2)$ in the Markov partition shown in Figure \ref{fig:Casdagli-Markov}. This set is densely intersected by the stable manifold on $\mathbb{T}^2$ of the point $(0,0)$, and these intersections are carried by $F$ to a dense subset of $\Lambda_1$ formed by intersections of $\Lambda_1$ with the strong-stable manifold on $\mathbb{S}$ of the point $P_1$.

Let $\vartheta$ be the curve of period-two periodic points for $f$, passing through $P_1$, as defined in \eqref{eq_per_pnts}. Recall that $\mathbb{W}^\mathrm{s}_1$ denotes the stable manifold to $\vartheta$. Since $\vartheta$ has two smooth branches connecting at $P_1$, $\mathbb{W}^\mathrm{s}_1$ can be realized as two smooth manifolds, call them $\mathbb{W}^\mathrm{s}_{1,j}$, $j=1,2$, that connect smoothly along the strong-stable manifold of $P_1$ on $S_0$.

\begin{lem}\label{lem:thm-i-helper1}
Let $\Lambda_r$ be a compact line segment along $\gamma_r$ which contains the intersection of $\gamma_r$ with the center-stable manifolds. Assume also that the endpoints of $\Lambda_r$ belong to this intersection. Then for all $r\approx 1$ but not equal to one, there exists a set $\set{G^i_r}_{i\in\N}$ of open, mutually disjoint subintervals of $\Lambda_r$ (we call them gaps), such that $\widetilde{B}_\infty\subset\Lambda_r\setminus\bigcup_i G^i_r$, and the collection of endpoints of all $G^i_r$ is a dense subset of $\widetilde{B}_\infty$. Moreover, for each $i$, one of the endpoints of $G^i_r$ belongs to $\mathbb{W}^\mathrm{s}_{1,1}$, and the other to $\mathbb{W}^\mathrm{s}_{1,2}$.
\end{lem}

\begin{proof}
Let $r_0\in(0,1)$, such that for all $r\in(1-r_0, 1+r_0)$, $r\neq 1$, $\Lambda_r$ intersects the center-stable manifolds transversely. Since $\Lambda_{r=1}$ intersects the strong-stable manifold of $P_1$ transversely (in a dense set of points), as soon as $r$ is slightly perturbed, a gap opens with one endpoint in $\mathbb{W}^\mathrm{s}_{1,1}$, the other in $\mathbb{W}^\mathrm{s}_{1,2}$ (see Figure \ref{fig:gap}). This gap persists for all $r\in(1-r_0, 1+r_0)$ (i.e. as long as $\Lambda_r$ intersects the center-stable manifolds transversely). 

In order to show that the endpoints of these gaps form a dense subset of $\widetilde{B}_\infty$, it is enough to show that no point inside of a gap belongs to $\widetilde{B}_\infty$. This follows from, for example, \cite[Theorem 5.22]{Cantat2009} (in fact, the strong-stable and strong-unstable manifolds of the eight points that are born from the singularities $\set{P_1,\dots,P_4}$ form boundaries of the Markov partition on $S_V$, $V > 0$--see also \cite{Casdagli1986}). 
\end{proof}

\begin{figure}[t]
\begin{minipage}[b]{0.4\linewidth}
\centering
\includegraphics[scale=.35]{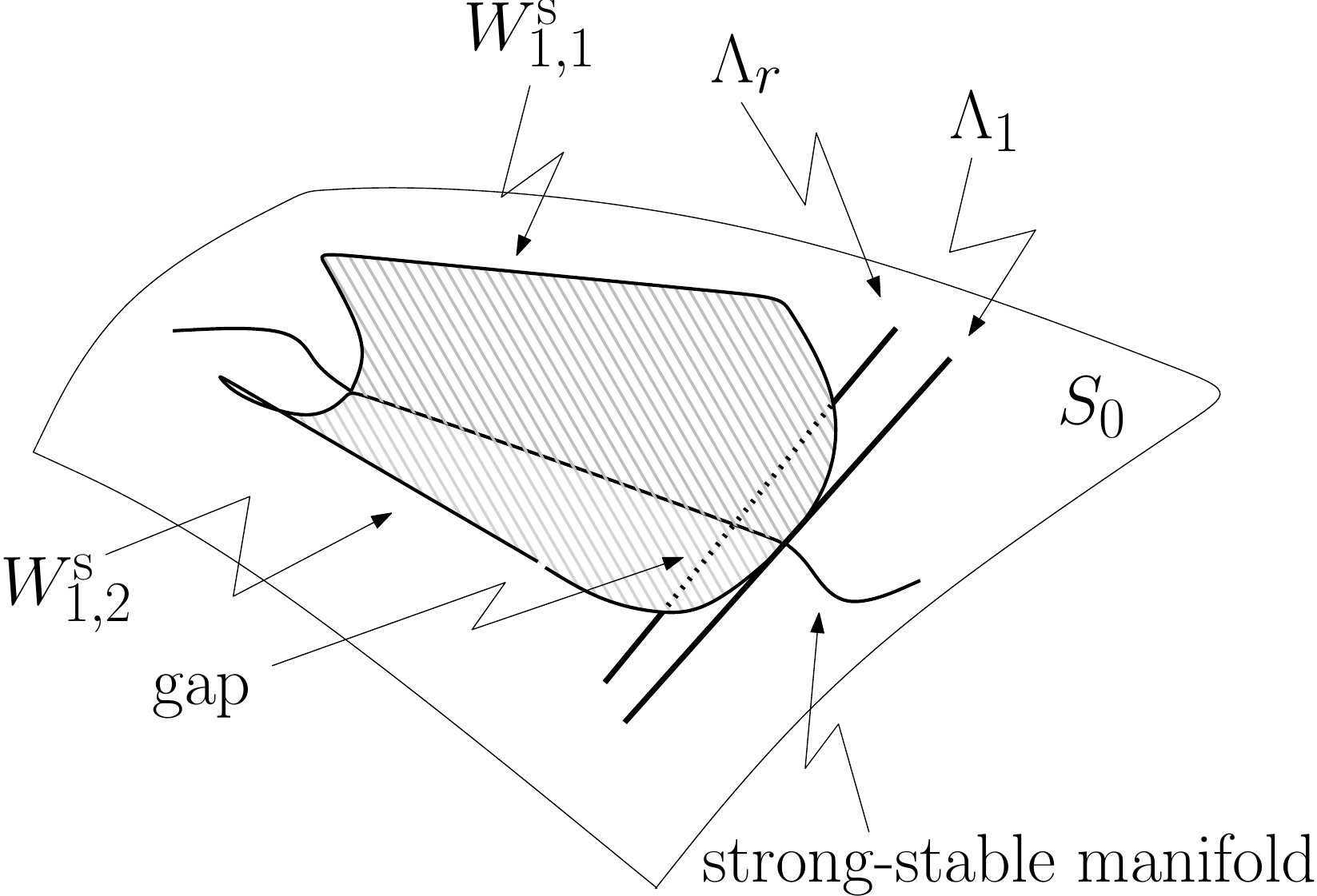}
\caption{}\label{fig:gap}
\end{minipage}
\hspace{6mm}
\begin{minipage}[b]{0.5\linewidth}
\centering
\includegraphics[scale=.35]{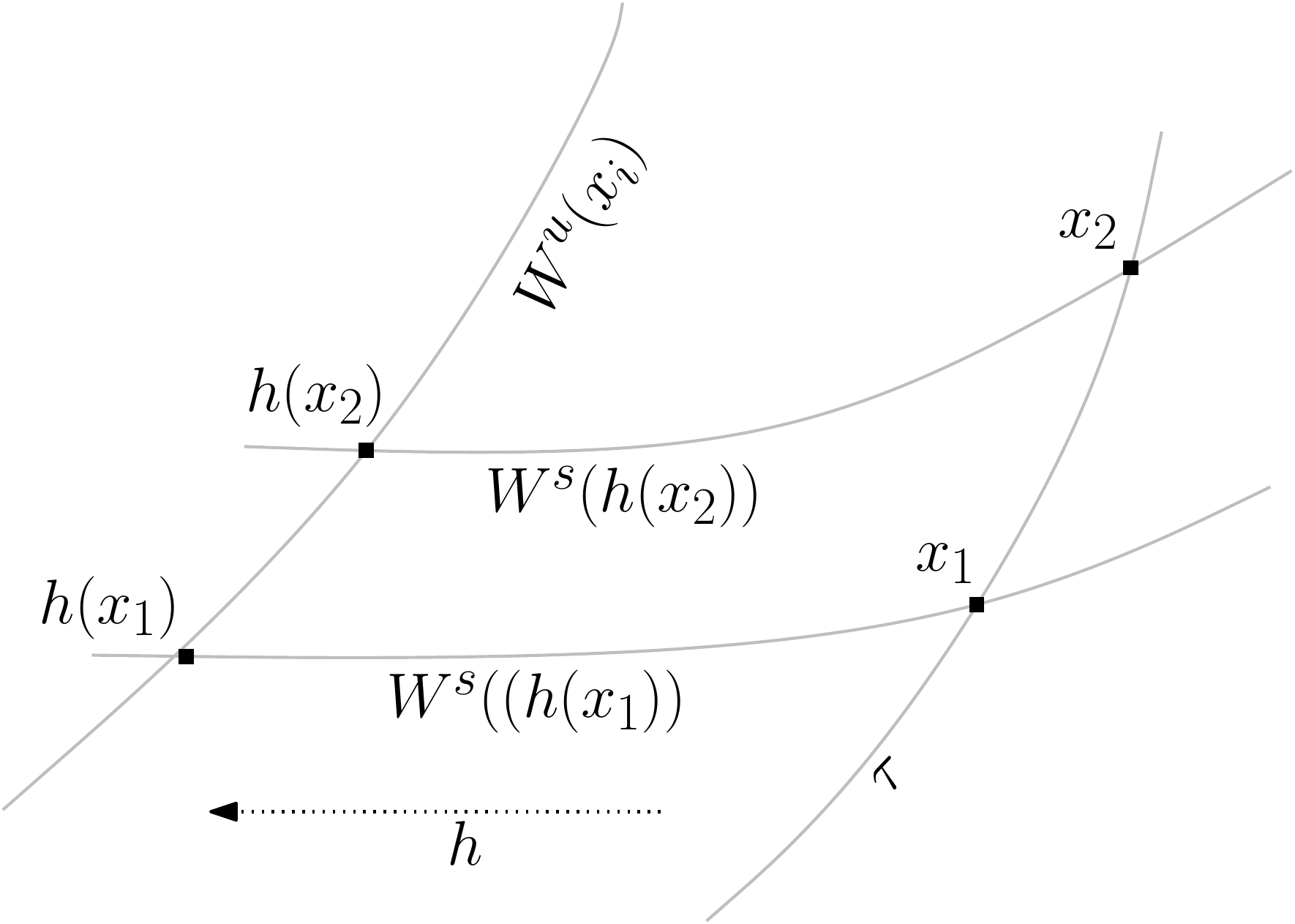}
\caption{}
\label{fig_holonomy}
\end{minipage}
\end{figure}

Fix $\epsilon > 0$ and let $C_1,\dots,C_m$ be an open cover of $\widetilde{B}_\infty$, with $\diam(C_i)\leq \epsilon$, $i = 1,\dots, m$. It follows that for all $k$ sufficiently large, $\Sigma_k$ is entirely contained in $\bigcup_i C_i$.

Now, for any $j\in\set{1,\dots,m}$, pick a gap whose endpoints lie inside of $C_j$. Call one endpoint $e_1$, and the other $e_2$. Assume that $e_1$ lies on $S_{V_1}$, and $e_2$ on $S_{V_2}$ (of course, $V_1, V_2 > 0$). Let $\vartheta\cap S_{V_1} = \set{p_1,q_1}$, and $\vartheta\cap S_{V_2}=\set{p_2,q_2}$ (here $f(p_i) = q_i$, $f(q_i) = p_i$). Say $p_1$ and $p_2$ lie on the same one of the two branches of $\vartheta$, and $e_1\in W^\mathrm{s}(p_1)$. Then $e_2\in W^\mathrm{s}(q_2)$. 

Observe that if $\mathbf{x} = (x, x/(2x - 1), x)\in\vartheta$, then $f(\mathbf{x}) = (x/(2x - 1), x, x/(2x-1))$. Hence if $\mathbf{x}\neq P_1$, either $\abs{\pi(\mathbf{x})} < 1$ or $\abs{\pi\circ f(\mathbf{x})}<1$. It follows that $\abs{\pi(p_i)}<1$ or $\abs{\pi(q_i)}<1$. Since
\begin{align*}
\abs{f^k(p_1) - f^k(e_1)}\xrightarrow[k\rightarrow\infty]{}0\text{\hspace{2mm} and \hspace{2mm}}\abs{f^k(q_2) - f^k(e_2)}\xrightarrow[k\rightarrow\infty]{}0,
\end{align*}
and $\abs{V_1 - V_2}$ is small provided that $\epsilon$ is sufficiently small (hence $\abs{p_1 - p_2}$ and $\abs{q_1 - q_2}$ are small), it follows that, for all $k$ sufficiently large, either $\abs{\pi\circ f^k(e_1)}<1$ or $\abs{\pi\circ f^k(e_2)}<1$. Therefore, for all $k$ sufficiently large, $\widetilde{\sigma}_{N+k}\cap C_j\neq \emptyset$, proving \eqref{eq:part2_cor10}.
\end{proof}
Now define
\begin{align*}
\widetilde{B}_\infty^{-} = \set{\lambda\in(-\infty, 0]:\abs{\lambda}\in \widetilde{B}_\infty} \text{\hspace{2mm} and \hspace{2mm}}
\widetilde{\sigma}_k^{-} = \set{\lambda\in(-\infty, 0]: \abs{\lambda}\in\widetilde{\sigma}_k}.
\end{align*}
Take $B_\infty = \widetilde{B}_\infty\cup\widetilde{B}_\infty^{-}$. Then $\sigma_k$ in \eqref{model_eq6} is precisely $\widetilde{\sigma}_k\cup\widetilde{\sigma}_k^{-}$, and so 
\begin{align*}
\sigma_k\xrightarrow[k\rightarrow\infty]{}B_\infty\text{\hspace{2mm}in Hausdorff metric.}
\end{align*}
This completes the proof of (i) of Theorem \ref{thm:main}. In what follows, we prove properties of $B_\infty$ stated in (ii)--(iv) of Theorem \ref{thm:main}.

\subsection{Proof of Theorem \ref{thm:main}-(ii)}\label{sec:proof-main-ii}

This is a direct consequence of Propositions \ref{part2_cor9} and \ref{part2_prop11}.

\subsection{Proof of Theorem \ref{thm:main}-(iii)}\label{sec:proof-main-iii}

Let $M$ be a smooth two-dimensional Riemannian manifold and $\Lambda\subset M$ a basic set for $f\in\diff^2(M)$. Let $\set{f_\alpha}\subset\diff^2(M)$ depending continuously on $\alpha\in\R^+$, with $f = f_0$. Then there exists $\beta > 0$ such that for all $\alpha\leq \beta$, $f_\alpha$ has a basic set $\Lambda_\alpha$ near $\Lambda_0 = \Lambda$. Let $\set{\tau_\alpha}$, $\alpha\in\R^+$, be a family of smooth compact regular curves depending continuously on $\alpha$ in the $C^1$-topology. Assume also that $\tau_0$ intersects $W^s(\Lambda_0)$ transversely. Then there exists $\beta\in\R^+$, such that for all $\alpha\in[0, \beta)$, $\tau_\alpha$ intersects $W^s(\Lambda_\alpha)$ transversely. Hence we may define the holonomy map
\begin{align}\label{holonomy_map}
h_\alpha: \tau_\alpha\cap W^s(\Lambda_\alpha)\rightarrow\Lambda_\alpha
\end{align}
by sliding points along the stable manifolds to an unstable one (see Figure \ref{fig_holonomy}). Then locally $h_\alpha^\loc$ and its inverse are well-defined, it is a homeomorphism onto its image, and, together with its inverse, is Lipschitz (Lipschitz continuity follows easily from Section \ref{a_1_3}).

\begin{prop}\label{lip_cont}
There exists $\beta > 0$ such that the Lipschitz constant for $h_\alpha^\loc$ and its inverse can be chosen uniformly for all $\alpha\in[0,\beta)$.
\end{prop}

To prove Proposition \ref{lip_cont}, one proceeds by a (rather standard) technique that was first introduced in \cite{Anosov1967}. The proof is sketched in Appendix \ref{c}.

Recall that a morphism of metric spaces $H:(M_1,d_1)\rightarrow (M_2,d_2)$ is said to be $(C,\nu)$-H\"older continuous provided that for all $x,y\in M_1$, $d_2(H(x),H(y))\leq Cd_1(x,y)^\nu$. We have the following, due to J. Palis and M. Viana \cite[Theorem B]{Palis1988}.
\begin{prop}\label{palis_viana}
Let $f:M\rightarrow M$ be a $C^1$ diffeomorphism on a Riemannian 2-manifold and $\Lambda\subset M$ a basic set for $f$ with $(1,1)$ splitting. Then there exists $C > 0$ and for any $\nu\in(0, 1)$ there exists $\mathcal{U}\subset \diff^1(M)$ an open neighborhood of $f$ such that for all $g\in\mathcal{U}$ and $x\in\Lambda$, $H_g|_{W^u(x)\cap \Lambda}$ and its inverse are $(C,\nu)$-H\"older continuous. Here $H_g:\Lambda\rightarrow\Lambda_g$ is the topological conjugacy (see Section \ref{b1-1}). 
\end{prop}
Under the hypothesis of and with the notation from Proposition \ref{lip_cont}, we get
\begin{lem}\label{holder_lemma}
Let $\beta > 0$ satisfying Proposition \ref{lip_cont}. There exists $C > 0$ and for any $\nu\in(0,1)$ there exists $\beta_0\in(0,\beta)$, such that for any $\alpha_1,\alpha_2\in [0, \beta_0)$, the map
\begin{align*}
[h_{\alpha_2}^\loc]^{-1}\circ H_{\alpha_1,\alpha_2}\circ h_{\alpha_1}^\loc: \tau_{\alpha_1}\cap W^s(\Lambda_{\alpha_1})\rightarrow \tau_{\alpha_2}\cap W^s(\Lambda_{\alpha_2}),
\end{align*}
where $H_{\alpha_1,\alpha_2}:\Lambda_{\alpha_1}\rightarrow\Lambda_{\alpha_2}$ is the topological conjugacy, is $(C,\nu)$-H\"older continuous (see Figure \ref{fig_palis_viana}); that is, we have the following diagram, with $(h_{\alpha_i}^{\loc})^{\pm 1}$ ($i = 1, 2$) Lipschitz, with the same Lipschitz constant, and $H_{\alpha_1,\alpha_2}$ H\"older continuous.
\begin{center}
\begin{tikzpicture}
 \matrix (m) [matrix of math nodes,
	      row sep = 3em,
	      column sep = 6em,
	      minimum width = 2em,text height=1.5ex, text depth=0.25ex]
 {
  \tau_1	&	W^u(y)\\
  \tau_2	&	H_{\alpha_1,\alpha_2}(W^u(y))\\
 };
 \path[-stealth]
  (m-1-1) edge node [above] {$h_{\alpha_1}$} (m-1-2)
  (m-1-2) edge node [right] {$H_{\alpha_1, \alpha_2}$} (m-2-2)
  (m-1-1) edge node [left]  {$h_{\alpha_2}^{-1}\circ H_{\alpha_1,\alpha_2}\circ h_{\alpha_1}$} (m-2-1)
  (m-2-2) edge node [below] {$h_{\alpha_2}^{-1}$} (m-2-1);
\end{tikzpicture}
\end{center}
where $y$ is such that $W^u(y)$ is an unstable manifold at the point $y$ containing the image of $h_{\alpha_1}$.
\end{lem}
\begin{rem}
 We cannot expect any higher modulus of continuity than H\"older in general. Indeed, we cannot expect the Hausdorff dimension of the hyperbolic sets of a family of diffeomorphisms to be constant (which would be implied if above instead of H\"older continuity we had Lipschitz); very simiple examples can be easily constructed. Yet more generally, holonomy maps along the so-called \textit{center foliations} are quite bad: see J. Milnor's exposition of Katok's example of so-called \textit{Fubini nightmare} in \cite{Milnor1997}, as well as genericity results in \cite{Shub2000}.
\end{rem}

\begin{figure}[t]
\centering
\includegraphics[scale=.5]{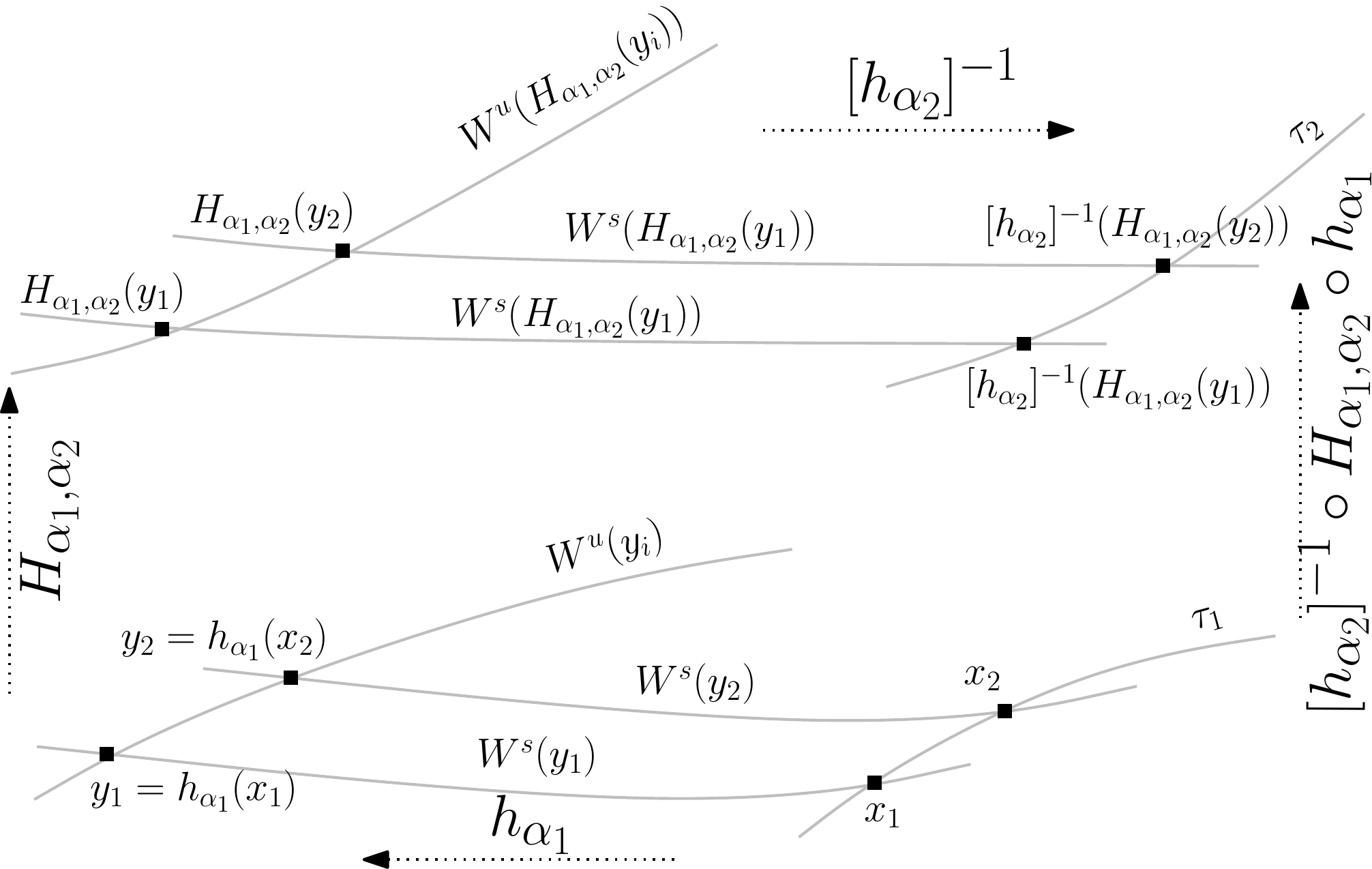}
\caption{}\label{fig_palis_viana}
\end{figure}
\begin{proof}
Let $x_1\in \tau_1\cap W^s(\Lambda_{\alpha_1})$ and $h_{\alpha_1}^\loc(x_1)\in W^u(y_1)$. Since the stable and unstable manifolds depend continuously on the point and on the diffeomorphism, if $\alpha_2$ is sufficiently close to $\alpha_1$, then $\tau_{\alpha_2}\cap W^s(H_{\alpha_1,\alpha_2}(y_1))\neq \emptyset$. Let $x_2\in W^s(H_{\alpha_1,\alpha_2}(y_1))\cap\tau_{\alpha_2}$ and $y_2 = H_{\alpha_1,\alpha_2}(y_1)$. Now $h_{\alpha_2}^\loc: U\rightarrow W^u(y_2)$, where $U$ is a neighborhood of $x_2$ in $\tau_{\alpha_2}$, is defined and Propositions \ref{lip_cont} and \ref{palis_viana} can be applied together.
\end{proof}

Recall from \eqref{eq:new_gamma}:
\begin{align*}
\gamma_r(\lambda) = \left(\frac{\lambda - (1 + J_1^2)}{2J_1},\frac{\lambda - (1 + r^2J_1^2)}{2rJ_1},\frac{1 + r^2}{2r}\right).
\end{align*}
Hence $\gamma_r$ lies in the plane
\begin{align*}
\Pi_r = \set{z = \frac{1 + r^2}{2r}}.
\end{align*}

Let $r_0 \in(0,1)$ be as in Proposition \ref{part2_prop11}. Let $r\in (1 - r_0, 1 + r_0)$, $r\neq 1$. Fix $x_0\in \gamma_r$ whose forward orbit under $f$ is bounded. Pick $\delta$ small, with $0 < \delta < V_0$, and let $\Gamma$ be a compact segment along $\gamma_r$ containing $x_0$, with endpoints lying on $S_{V_0 - \delta}$ and $S_{V_0 + \delta}$ (note, $x_0$ may be an endpoint of $\Gamma$). Then $\Gamma$ intersects the center-stable manifolds, as well as the surfaces $S_V$, $V\in [V_0 - \delta, V_0 + \delta]$, transversely. For $V\in [V_0 - \delta, V_0 + \delta]$, let $\tau_V$ denote the projection of $\Gamma$ onto $S_V$ along the plane $\Pi_r$. Then $\tau_V$ is a smooth, compact regular curve in $S_V$ intersecting $W^s(\Omega_V)$ transversely (transversality with the center-stable manifolds follows from Proposition \ref{part2_prop11}). Let $\mathcal{C} = \tau_{V_0}\cap W^s(\Omega_V)$. For every $x\in\mathcal{C}$, let
\begin{align*}
\vartheta_x = W^{cs}(x)\cap \Pi_r\cap\left(\bigcup_{V\in[V_0 - \delta, V_0 + \delta]}S_V\right),
\end{align*}
where $W^{cs}(x)$ is the center-stable manifold containing $x$. Then (see Figure \ref{fig_c1_c3})
\begin{enumerate}[\text{C}1.]
\item $\set{\tau_V}$ form a smooth foliation of $\Pi$;
\item $\vartheta_x$ is a smooth compact regular curve with endpoints lying in $S_{V_0 - \delta}$ and $S_{V_0 + \delta}$; 
\item $\set{\vartheta_x}_{x\in\mathcal{C}}$ intersects $\set{\tau_V}_{V\in[V_0 - \delta, V_0 + \delta]}$ and $\Gamma$ uniformly transversely, and the curves $\vartheta_x$ depend continuously on $x$ in the $C^1$-topology (see Remark \ref{part2_rem9});
\item For $\delta > 0$ sufficiently small, $\set{\tau_V}_{V\in[V_0 - \delta, V_0 + \delta]}$ and $\set{f_V}_{V\in[V_0 - \delta, V_0 + \delta]}$ satisfy the hypothesis of Lemma \ref{holder_lemma}. In particular, there exists $C > 0$ and for any $\nu \in (0,1)$ there exists $0 < \epsilon < \delta$, such that the map $\mathcal{C}\ni x\mapsto \tau_V\cap W^s(\Omega_V)$ defined by projecting points along the curves $\set{\vartheta_x}_{x\in\mathcal{C}}$ is $(C,\nu)$-H\"older continuous for all $V\in [V_0 - \epsilon, V_0 + \epsilon]$.
\end{enumerate}
\begin{figure}[t]
\centering
\includegraphics[scale=.5]{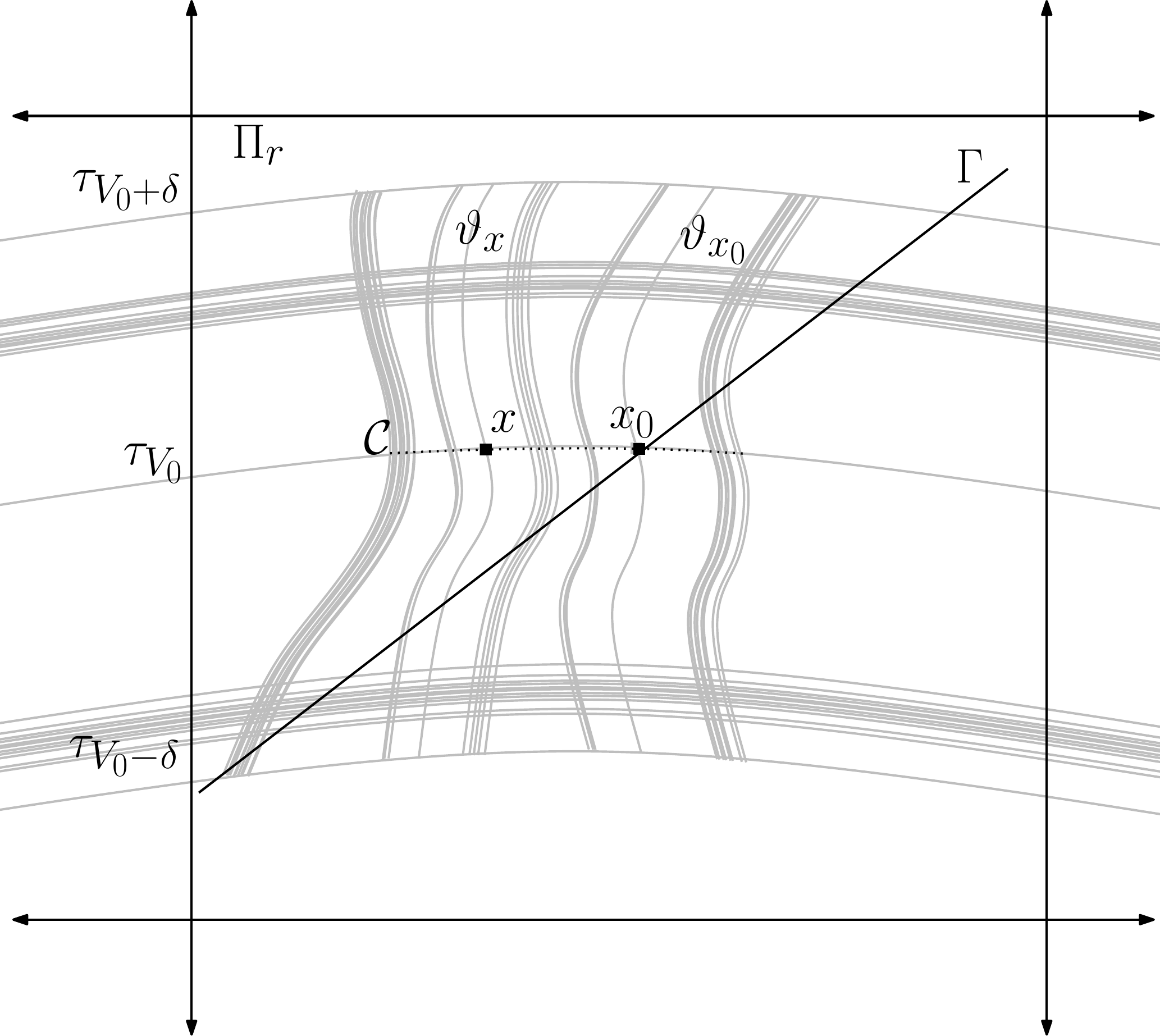}
\caption{}
\label{fig_c1_c3}
\end{figure}

Now, from C1--C4 it follows that there exists a sufficiently small neighborhood $U$ of $x_0$ in $\tau_{V_0}$ and $\widetilde{C} > 0$, such that the map $U\cap\mathcal{C} \ni x\mapsto \Gamma$ defined by projecting points along the curves $\set{\vartheta_x}_{x\in U\cap\mathcal{C}}$ is $(\widetilde{C},\nu)$-H\"older continuous (see \cite[Lemma 3.5]{Yessen2011} for technical details). Hence
\begin{align}\label{local_dim}
\lhdim(\Gamma,x_0) = \lhdim(\mathcal{C}, x_0).
\end{align}
On the other hand,
\begin{align*}
\lhdim(\mathcal{C},x_0) = \hdim(W_\loc^u(h(x_0))) = h^u(\Omega_{V_0}),
\end{align*}
where $h: \mathcal{C}\rightarrow \Omega_{V_0}$ is a holonomy map as defined in \eqref{holonomy_map}, and $h^u$ is defined in Section \ref{b1-4} as the Hausdorff dimension of $\Omega_V$ along leafs of the unstable lamination $W^\mathrm{u}(\Omega_V)$. Now, $h^u(\Omega_V)$ depends continuously (in fact analytically, as we shall see below) on $V$. It follows that the local Hausdorff dimension is continuous over $B_\infty$.

\begin{prop}[{\cite[Theorem 5.23]{Cantat2009}}]\label{cantat_prop}
Let $\gamma$ be an analytic curve in $\bigcup_{V> 0}S_V$ parameterized on $(0,1)$. Then $h^{s,u}\left(\Omega_{\gamma(t)}\right): (0,1)\rightarrow\R$ is analytic with values strictly between zero and one.
\end{prop}

\begin{prop}[{\cite[Theorem 1]{Damanik2009c}}]\label{damanik_gorodetski_prop}
The Hausdorff dimension of $\Omega_V$ is right-continuous at $V = 0$: $\lim_{V\rightarrow 0^+}h^u(\Omega_V) = 1$.
\end{prop}

Now, for $r\in (1 - r_0, 1 + r_0)$, let
\begin{align*}
\Gamma_r = \set{\lambda: \mathcal{O}^+_f(\gamma_r(\lambda))\text{ is bounded}}.
\end{align*}
As we have already seen, $\Gamma_r$ is a Cantor set; in particular it contains limit points. On the other hand, $\gamma_r$ is an analytic curve that intersects $S_V$ for every $V\in [0, \infty)$. Combining this with  \eqref{local_dim}, Propositions \ref{cantat_prop} and \ref{damanik_gorodetski_prop}, we get that the local Hausdorff dimension is non-constant over the spectrum.

Also, \eqref{local_dim} in combination with Proposition \ref{cantat_prop} shows that the local Hausdorff dimension at every point of $B_\infty$ is strictly between zero and one, hence so is the global Hausdorff dimension. It follows that the Lebesgue measure of $B_\infty$ is zero.

\subsection{Proof of Theorem \ref{thm:main}-iv}\label{sec:last}

Observe that the line $\gamma_{(J_0,J_1)}$ is continuous in the parameters $(J_0,J_1)$. From constructions carried out in the previous section, it is evident that $\hdim(B_\infty(J_0,J_1))$ is also continuous in the parameters $(J_0,J_1)$. Also, by Proposition \ref{damanik_gorodetski_prop}, continuity extends to the pure case $J_0 = J_1$.

\section*{Acknowledgement}
I wish to express gratitude to my dissertation advisor, Anton Gorodetski, who originally introduced me to this problem and whose guidance and support have been essential for completion of this project.

I also wish to thank Michael Baake, Jean Bellissard, Alexander Chernyshev, David Damanik, Uwe Grimm, Svetlana Jitomirskaya, Ron Lifshitz and Laurent Raymond for meaningful and illuminating discussions.

A special thanks to Michael Baake, David Damanik and Anton Gorodetski for taking the time to read and comment on a draft.

Finally, I wish to thank the anonymous referees for their very helpful suggestions and remarks.

\appendixpage

\appendix

None of what is presented in the following appendices is new. In particular, all of appendix \ref{b} is by now part of the classical theory of hyperbolic and partially hyperbolic dynamical systems (and references are given to comprehensive surveys). 

\section{Normal hyperbolicity of six-cycles through singularities of \texorpdfstring{$S_0$}{} and symmetries of \texorpdfstring{$f$}{}}\label{a}

\subsection{Normal hyperbolicity of six-cycles through singularities of \texorpdfstring{$S_0$}{}}\label{a1}

As has been mentioned above, $S_0$ contains four conic singularities; explicitly, they are 
\begin{align}\label{eq:singularities} 
P_1 = (1,1,1),\hspace{2mm} P_2 = (-1,-1,1),\hspace{2mm} P_3 = (1,-1,-1),\hspace{2mm} P_4 = (-1,1,-1). 
\end{align} 
The point $P_1$ is fixed under $f$, while $P_2$, $P_3$ and $P_4$ form a three cycle: 
\begin{align*} 
P_1\overset{f}{\longmapsto}P_1;\hspace{4mm} P_2\overset{f}{\longmapsto}P_3\overset{f}{\longmapsto}P_4\overset{f}{\longmapsto}P_2 
\end{align*} 
(which can be verified via direct computation).

For each $i\in\set{1,\dots,4}$, there is a smooth curve $\rho_i$ which does not contain any self-intersections, passing through the singularity $P_i$, such that $\rho_i\setminus{P_i}$ is a disjoint union of two smooth curves---call them $\rho_i^l$ and $\rho_i^r$---with the following properties: 
\begin{itemize} 

\item $\rho_i^{l,r}\subset \bigcup_{V>0}S_V^\R$; 

\item $f(\rho_1^l) = \rho_1^r$ and $f(\rho_1^r) = f(\rho_1^l)$. In particular, points of $\rho_1^{l,r}$ are periodic of period two, and $\rho_1$ is fixed under $f$; 

\item The six curves $\rho_i^{l,r}$, $i = 2, 3, 4$, form a six cycle under $f$. In particular, points of $\rho_i^{l,r}$, $i = 2, 3, 4$, are periodic of period six, and hence for $i = 2, 3, 4$, $\rho_i$ is fixed under $f^6$. 

\end{itemize}
The curve $\rho_1$ is given explicitly by
\begin{align}\label{eq:per-2-curve} 
\rho_1 = \set{\left(x,\hspace{1mm}\frac{x}{2x-1},\hspace{1mm}x\right): x\in\left(-\infty,1/2\right)\cup\left(1/2,\infty\right)}. 
\end{align} 
Expressions for the other three curves can be obtained from \eqref{eq:per-2-curve} using symmetries of $f$ to be discussed below.  

It follows via a simple computation that for any $i \in \set{1,\dots, 4}$ and any point $p\in \rho_i^{l,r}$, the eigenvalue spectrum of $Df^6_p$ is $\set{1, \lambda(p), 1/\lambda(p)}$ with $0 < \abs{\lambda(p)} < 1$, where $Df^6_p$ denotes the differential of $f^6$ at the point $p$. The eigenspace corresponding to the eigenvalue $1$ is tangent to $\rho_i$ at $p$. At $P_i$, the eigenvalue spectrum of $Df^6_{P_i}$ is of the same form, and as above, the eigenspace corresponding to the unit eigenvalue is tangent to $\rho_i$ at $P_i$. It follows that the curves $\set{\rho_i}_{i\in\set{1,\dots,4}}$ are normally hyperbolic one-dimensional submanifolds of $\R^3$, as defined in Section \ref{b3}.

For $V > 0$, $S_V^\R\bigcap (\rho_1\cup\cdots\cup\rho_4)$ consists of eight hyperbolic periodic points for $f$ which are fixed hyperbolic points for $f^6$. The stable manifolds to these points form a dense sublamination of the stable lamination
\begin{align*}
W^s(\Omega_V) = \bigcup_{x\in\Omega_V}W^s(x),
\end{align*}
the union of global stable manifolds to points in $\Omega_V$, the nonwandering set for $f$ on $S_V$ from Theorem \ref{part2_thm1}. Moreover, this sublamination forms the boundary of the stable lamination. For details, see \cite{Casdagli1986, Damanik2009, Cantat2009}. This extends to the center-stable lamination: the stable manifolds to the normally hyperbolic curves $\rho_1,\dots,\rho_4$ form a dense sublamination of the (two-dimensional) center-stable lamination and forms the boundary of this lamination. In particular, if $\gamma$ is a smooth curve intersecting a center-stable manifold, and this intersection is not isolated, then in an arbitrarily small neighborhood of this intersection, $\gamma$ intersects the stable manifolds of all eight curves: $\set{\rho_i^{l,r}}_{i\in\set{1,\dots,4}}$.

\subsection{Symmetries of \texorpdfstring{$f$}{}}\label{a2}

The following discussion is taken from \cite{Damanik0000my1}; however, what follows does not appear in \cite{Damanik0000my1} as new results but a recollection of what is known. In particular, the reader should consult \cite{Baake1997, Roberts1994} and references therein, as well as earlier (and original) works \cite{Kohmoto1983, Kohmoto1992, Kadanoff0000, Kadanoff1984}.

Let us denote the group of symmetries of $f^6$ by $\mathcal{G}_\mathrm{sym}$, and the group of reversing symmetries of $f^6$ by $\mathcal{G}_\mathrm{rev}$; that is, 
\begin{align}\label{eq:sym-group} 
\mathcal{G}_\mathrm{sym} = \set{s\in\mathrm{Diff}(\R^3): s\circ f^6\circ s^{-1} = f^6}, 
\end{align} 
and 
\begin{align}\label{eq:rev-group} 
\mathcal{G}_\mathrm{rev} = \set{s\in\mathrm{Diff}(\R^3): s\circ f^6\circ s^{-1} = f^{-6}}, 
\end{align} 
where $\mathrm{Diff}(\R^3)$ denotes the set of diffeomorphisms on $\R^3$. 

Observe that $\mathcal{G}_\mathrm{rev}\neq\emptyset$. Indeed,  
\begin{align}\label{eq:rev-sym} 
s(x,y,z) = (z,y,x) 
\end{align} 
is a reversing symmetry of $f$, and hence also of $f^6$. Hence $f^6$ is smoothly conjugate to $f^{-6}$. It follows (see Appendix \ref{a1}) that forward-time dynamical properties of $f^6$, as well as the geometry of dynamical invariants (such as stable manifolds) are mapped smoothly and rigidly to those of $f^{-6}$. That is, forward-time dynamics of $f^6$ is essentially the same as its backward-time dynamics. 

The group $\mathcal{G}_\mathrm{sym}$ is also nonempty, and more importantly, it contains the following diffeomorphisms: 
\begin{align}\label{eq:symmetries} 
s_2: (x,y,z)\mapsto (-x, -y, z),\notag\\ 
s_3: (x,y,z)\mapsto(x,-y,-z),\\ 
s_4: (x,y,z)\mapsto(-x,y,-z).\notag 
\end{align} 
Notice that $s_i$ are rigid transformations. Also notice that 
\begin{align*}
s_i(P_1) = P_i,
\end{align*} 
and since $s_i$ is in fact a smooth conjugacy, we must have
\begin{align}\label{eq:symmetries-on-rho}
s_i(\rho_1) = \rho_i.
\end{align}
For a more general and extensive discussion of symmetries and reversing symmetries of trace maps, see \cite{Baake1997}.

\section{Background on uniform, partial and normal hyperbolicity}\label{b}

\subsection{Properties of locally maximal hyperbolic sets}\label{b1}

A more detailed discussion can be found in \cite{Hirsch1968, Hirsch1970, Hirsch1977, Hasselblatt2002b, Hasselblatt2002}.

A closed invariant set $\Lambda\subset M$ of a diffeomorphism $f: M\rightarrow M$ of a smooth manifold $M$ is called \textit{hyperbolic} if for each $x\in\Lambda$, there exists the splitting $T_x\Lambda = E_x^s\oplus E_x^u$ invariant under the differential $Df$, and $Df$ exponentially contracts vectors in $E_x^s$ and exponentially expands vectors in $E_x^u$. The set $\Lambda$ is called \textit{locally maximal} if there exists a neighborhood $U$ of $\Lambda$ such that
\begin{align}\label{part2_eq1}
\Lambda = \bigcap_{n\in\Z}f^n(U).
\end{align}
The set $\Lambda$ is called \textit{transitive} if it contains a dense orbit. It isn't hard to prove that the splitting $E_x^s\oplus E_x^u$ depends continuously on $x\in\Lambda$, hence $\dim(E_x^{s,u})$ is locally constant. If $\Lambda$ is transitive, then $\dim(E_x^{s,u})$ is constant on $\Lambda$. We call the splitting $E_x^s\oplus E_x^u$ a $(k_x^s, k_x^u)$ splitting if $\dim(E_x^{s,u}) = k^{s,u}$, respectively. In case $\Lambda$ is transitive, we shall simply write $(k^s, k^u)$. 
\begin{defn}\label{basic_set}
We call $\Lambda\subset M$ a \textit{basic set} for $f\in\diff^r(M)$, $r\geq 1$, if $\Lambda$ is a locally maximal invariant transitive hyperbolic set for $f$.
\end{defn}
Suppose $\Lambda$ is a basic set for $f$ with $(1,1)$ splitting. Then the following holds.

\subsubsection{Stability}\label{b1-1}

Let $U$ be as in \eqref{part2_eq1}. Then there exists $\mathcal{U}\subset \diff^1(M)$ open, containing $f$, such that for all $g\in\mathcal{U}$,
\begin{align}\label{part2_eq2}
\Lambda_g = \bigcap_{n\in\Z}g^n(U)
\end{align}
is $g$-invariant transitive hyperbolic set; moreover, there exists a (unique) homeomorphism $H_g:\Lambda\rightarrow\Lambda_g$ such that
\begin{align}\label{part2_eq3}
H_g\circ f|_{\Lambda} = g|_{\Lambda_g}\circ H_g;
\end{align}
that is, the following diagram commutes.
\begin{center}
\begin{tikzpicture}
 \matrix (m) [matrix of math nodes,
	      row sep = 3em,
	      column sep = 6em,
	      minimum width = 2em,text height=1.5ex, text depth=0.25ex]
 {
  \Lambda	&	\Lambda\\
  \Lambda_g	&	\Lambda_g\\
 };
 \path[-stealth]
  (m-1-1) edge node [above] {$f$} (m-1-2)
  (m-1-2) edge node [right] {$H_g$} (m-2-2)
  (m-1-1) edge node [left]  {$H_g$} (m-2-1)
  (m-2-1) edge node [below] {$g$} (m-2-2);
\end{tikzpicture}
\end{center}

Also $H_g$ can be taken arbitrarily close to the identity by taking $\mathcal{U}$ sufficiently small. In this case $g$ is said to be \textit{conjugate to} $f$, and $H_g$ is said to be \textit{the conjugacy}.

\subsubsection{Stable and unstable invariant manifolds}\label{b1-2}

Let $\epsilon > 0$ be small. For each $x\in\Lambda$ define the \textit{local stable} and \textit{local unstable} manifolds at $x$:
\begin{align*}
W_\epsilon^s(x) = \set{y\in M: d(f^n(x),f^n(y))\leq \epsilon \text{ for all }n\geq 0},
\end{align*}
\begin{align*}
W_\epsilon^u(x) = \set{y\in M: d(f^n(x),f^n(y))\leq \epsilon \text{ for all }n\leq 0}.
\end{align*}
We sometimes do not specify $\epsilon$ and write
\begin{align*}
W_\loc^s(x)\text{\hspace{5mm}and\hspace{5mm}}W_\loc^u(x)
\end{align*}
for $W_\epsilon^s(x)$ and $W_\epsilon^u(x)$, respectively, for (unspecified) small enough $\epsilon > 0$. For all $x\in\Lambda$, $W_\loc^{s,u}(x)$ is an embedded $C^r$ disc with $T_xW_\loc^{s,u}(x) = E_x^{s,u}$. The \textit{global stable} and \textit{global unstable} manifolds
\begin{align}\label{part2_eq4}
W^s(x) = \bigcup_{n\in\N}f^{-n}(W_\loc^s(x))\text{\hspace{5mm}and\hspace{5mm}}W^u(x) = \bigcup_{n\in\N}f^{n}(W_\loc^u(x))
\end{align}
are injectively immersed $C^r$ submanifolds of $M$. Define also the stable and unstable sets of $\Lambda$:
\begin{align}\label{part2_eq5}
W^s(\Lambda) = \bigcup_{x\in\Lambda}W^s(x)\text{\hspace{5mm}and\hspace{5mm}}W^u(\Lambda) = \bigcup_{x\in\Lambda}W^u(x).
\end{align}

If $\Lambda$ is compact, there exists $\epsilon > 0$ such that for any $x,y\in\Lambda$, $W_\epsilon^s(x)\cap W_\epsilon^u(y)$ consists of at most one point, and there exists $\delta > 0$ such that whenever $d(x,y) < \delta$, $x,y\in\Lambda$, then $W_\epsilon^s(x)\cap W_\epsilon^u(y)\neq \emptyset$. If in addition $\Lambda$ is locally maximal, then $W_\epsilon^s(x)\cap W_\epsilon^u(y)\in\Lambda$. 

The stable and unstable manifolds $W_\loc^{s,u}(x)$ depend continuously on $x$ in the sense that there exists $\Phi^{s,u}:\Lambda\rightarrow\emb^r(\R, M)$ continuous, with $\Phi^{s,u}(x)$ a neighborhood of $x$ along $W_\loc^{s,u}(x)$, where $\emb^r(\R, M)$ is the set of $C^r$ embeddings of $\R$ into $M$ \cite[Theorem 3.2]{Hirsch1968}.

The manifolds also depend continuously on the diffeomorphism in the following sense. For all $g\in\diff^r(M)$ $C^r$ close to $f$, define
$\Phi_g^{s,u}:\Lambda_g\rightarrow\emb^r(\R,M)$ as we defined $\Phi^{s,u}$ above. Then define
\begin{align*}
\tilde{\Phi}_g^{s,u}:\Lambda\rightarrow\emb^r(\R, M)
\end{align*}
by
\begin{align*}
\tilde{\Phi}_g^{s,u} = \Phi_g^{s,u}\circ H_g.
\end{align*}
Then $\tilde{\Phi}^{s,u}_g$ depends continuously on $g$ \cite[Theorem 7.4]{Hirsch1968}.

\subsubsection{Fundamental domains}\label{b1-3}

Along every stable and unstable manifold, one can construct the so-called \textit{fundamental domains} as follows. Let $W^s(x)$ be the stable manifold at $x$. Let $y\in W^s(x)$. We call the arc $\gamma$ along $W^s(x)$ with endpoints $y$ and $f^{-1}(y)$ a \textit{fundamental domain}. The following holds.
\begin{itemize}
\item $f(\gamma)\cap W^s(x) = y$ and $f^{-1}(\gamma)\cap W^s(x) = f^{-1}(y)$, and for any $k\in\Z$, if $k < -1$, then $f^k(\gamma)\cap W^s(x) = \emptyset$; if $k > 1$ then $f^k(\gamma)\cap W^s(x) = \emptyset$ iff $x\neq y$;
\item For any $z\in W^s(x)$, if for some $k\in\N$, $f^k(z)$ lies on the arc along $W^s(x)$ that connects $x$ and $y$, then there exists $n\in\N$, $n\leq k$, such that $f^n(z)\in\gamma$.
\end{itemize}
Similar results hold for the unstable manifolds.

\subsubsection{Invariant foliations}\label{a_1_3}

A stable foliation for $\Lambda$ is a foliation $\mathcal{F}^s$ of a neighborhood of $\Lambda$ such that
\begin{enumerate}
\item for each $x\in\Lambda$, $\mathcal{F}(x)$, the leaf containing $x$, is tangent to $E_x^s$;
\item for each $x$ sufficiently close to $\Lambda$, $f(\mathcal{F}^s(x))\subset\mathcal{F}^s(f(x))$.
\end{enumerate}
An unstable foliation $\mathcal{F}^u$ is defined similarly.

For a locally maximal hyperbolic set $\Lambda\subset M$ for $f\in\diff^1(M)$, $\dim(M) = 2$, stable and unstable $C^0$ foliations with $C^1$ leaves can be constructed; in case $f\in\diff^2(M)$, $C^1$ invariant foliations exist (see \cite[Section A.1]{Palis1993} and the references therein).

\subsubsection{Local Hausdorff and box-counting dimensions}\label{b1-4}

For $x\in\Lambda$ and $\epsilon > 0$, consider the set $W_\epsilon^{s,u}\cap\Lambda$. Its Hausdorff dimension is independent of $x\in\Lambda$ and $\epsilon > 0$.

Let
\begin{align}\label{part2_eq6}
h^{s,u}(\Lambda) = \dim_H(W_\epsilon^{s,u}(x)\cap\Lambda).
\end{align}
For properly chosen $\epsilon > 0$, the sets $W_\epsilon^{s,u}(x)\cap\Lambda$ are dynamically defined Cantor sets, so
\begin{align*}
h^{s,u}(\Lambda) < 1
\end{align*}
(see \cite[Chapter 4]{Palis1993}). Moreover, $h^{s,u}$ depends continuously on the diffeomorphism in the $C^1$-topology \cite{Manning1983}. In fact, when $\dim(M) = 2$, these are $C^{r-1}$ functions of $f\in\diff^r(M)$, for $r\geq 2$ \cite{Mane1990}.

Denote the box-counting dimension of a set $\Gamma$ by $\dim_{\mathrm{Box}}(\Gamma)$. Then
\begin{align*}
\dim_H(W^{s,u}_\epsilon(x)\cap \Lambda) = \dim_{\mathrm{Box}}(W^{s,u}_\epsilon(x)\cap \Lambda)
\end{align*}
(see \cite{Manning1983, Takens1988}).

\subsection{Partial hyperbolicity}\label{b2}

For a more detailed discussion, see \cite{Pesin2004, Hasselblatt2006}.

An invariant set $\Lambda\subset M$ of a diffeomorphism $f\in\diff^r(M)$, $r\geq 1$, is called \textit{partially hyperbolic (in the narrow sense)} if for each $x\in\Lambda$ there exists a splitting $T_xM = E_x^s\oplus E_x^c\oplus E_x^u$ invariant under $Df$, and $Df$ exponentially contracts vectors in $E_x^s$, exponentially expands vectors in $E_x^u$, and $Df$ may contract or expand vectors from $E_x^c$, but not as fast as in $E_x^{s,u}$. We call the splitting $(k_x^s, k_x^c, k_x^u)$ splitting if $\dim(E_x^{s,c,u}) = k_x^{s,c,u}$, respectively. We shall write $(k^s,k^c,k^u)$ if the dimension of subspaces does not depend on the point.

\subsection{Normal hyperbolicity}\label{b3}

For a more detailed discussion and proofs see \cite{Hirsch1977} and also \cite{Pesin2004}.

Let $M$ be a smooth Riemannian manifold, compact, connected and without boundary. Let $f\in\diff^r(M)$, $r\geq 1$. Let $N$ be a compact smooth submanifold of $M$, invariant under $f$. We call $f$ \textit{normally hyperbolic} on $N$ if $f$ is partially hyperbolic on $N$. That is, for each $x\in N$,
\begin{align*}
T_xM = E_x^s\oplus E_x^c\oplus E_x^u
\end{align*}
with $E_x^c = T_xN$. Here $E_x^{s,c,u}$ is as in Section \ref{b2}. Hence for each $x\in N$ one can construct local stable and unstable manifolds $W_\epsilon^s(x)$ and $W_\epsilon^u(x)$, respectively, such that
\begin{enumerate}
\item $x\in W_\loc^s(x)\cap W_\loc^u(x)$;
\item $T_x W_\loc^s(x) = E^s(x)$, $T_xW_\loc^u(x) = E^u(x)$;
\item for $n\geq 0$,
\begin{align*}
d(f^n(x),f^n(y))\xrightarrow[n\rightarrow\infty]{}0\text{ for all }y\in W_\loc^s(x),
\end{align*}
\begin{align*}
d(f^{-n}(x),f^{-n}(y))\xrightarrow[n\rightarrow\infty]{}0\text{ for all }y\in W_\loc^u(x).
\end{align*}
\end{enumerate}
(For the proof see \cite[Theorem 4.3]{Pesin2004}).
These can then be extended globally by
\begin{align}\label{eq:ss-manifold}
W^s(x)& = \bigcup_{n\in\N}f^{-n}(W^s_\loc(x));\\
W^u(x)& = \bigcup_{n\in\N}f^n(W^u_\loc(x)).
\end{align}

The manifold $W^s(x)$ is referred to as the \textit{strong-stable manifold}, while $W^u(x)$ is called the \textit{strong-unstable} manifold; sometimes to emphasize the point $x$, we add \textit{ at $x$}.

Set
\begin{align}\label{eq:cscu-def}
W_\loc^{cs}(N) = \bigcup_{x\in N}W_\loc^s(x)\text{\hspace{5mm}and\hspace{5mm}} W_\loc^{cu}(N) = \bigcup_{x\in N}W_\loc^u(x).
\end{align}
\begin{thm}[Hirsch, Pugh and Shub \cite{Hirsch1977}]\label{a1_thm1}
The sets $W_\loc^{cs}(N)$ and $W_\loc^{cu}(N)$, restricted to a neighborhood of $N$, are smooth submanifolds of $M$. Moreover,
\begin{enumerate}
\item $W^\mathrm{cs}_\loc(N)$ is $f$-invariant and $W^\mathrm{cu}_\loc$ is $f^{-1}$-invariant;
\item $N = W^\mathrm{cs}_\loc(N)\bigcap W^\mathrm{cu}_\loc(N)$;
\item For every $x\in N$, $T_x W^\mathrm{cs, cu}_\loc(N) = E_x^{s,u}\oplus T_xN$;
\item $W^\mathrm{cs}_\loc(N)$ ($W^\mathrm{cu}_\loc(N)$) is the only $f$-invariant ($f^{-1}$-invariant) set in a neighborhood of $N$;
\item $W^\mathrm{cs}_\loc(N)$ (respectively, $W^\mathrm{cu}_\loc(N)$) consists precisely of those points $y\in M$ such that for all $n\geq 0$ (respectively, $n\leq 0$), $d(f^n(x),f^n(y)) < \epsilon$ for some $\epsilon > 0$.
\item $W^\mathrm{cs,cu}_\loc(N)$ is foliated by $\set{W_\loc^\mathrm{s,u}(x)}_{x\in N}$.
\end{enumerate}
\end{thm}

\section{Background results}\label{c}

We prove here some background results that follow from rather general principles in dynamical systems.

\subsection{Proof of Proposition \ref{lip_cont}: sketch of main ideas}
To prove Proposition \ref{lip_cont}, one proceeds by a (rather standard) technique that was first introduced in \cite{Anosov1967}. Let us sketch the proof below.

First suppose $\tau_\alpha = \tau_0 = \tau$ for all $\alpha$. Let $x_0\in \tau\cap W^s(\Lambda)$ and $\gamma$ an open arc along $\tau$ containing $x_0$ such that $h^\loc$ and its inverse are defined along $\gamma$. Let $U$ be an open neighborhood of $\Lambda$ such that $\Lambda$ is maximal in $U$, and $U$ can be foliated into stable and unstable foliations. There exists $k_0\in \N$ such that for all $\alpha$, $f_\alpha^{k_0}(x_0)\in U$. Assuming $\gamma$ is sufficiently short, we also have $f_\alpha^{k_0}(\gamma)\subset U$.

To simplify notation, let us write $f$ for $f_0$. Let $\gamma^s = f^{k_0}(\gamma)$ and let $\gamma^u$ be the unstable manifold such that $f^{k_0}\circ h(x_0)\in \gamma^u$. Let $\tilde{h}: \gamma^s\rightarrow\gamma^u$ be the induced holonomy map:
\begin{align*}
\tilde{h}(x) = f^{k_0}\circ h^\loc\circ f^{-k_0}.
\end{align*}
By the $C^1$ stable foliation, $\tilde{h}$ may be considered as the restriction of a $C^1$ map $F: \gamma^s\rightarrow\gamma^u$ to the set $f^{k_0}(\gamma\cap W^s(\Lambda))$. Then for all $x,y\in\gamma^s$ sufficiently close, there exists $k = k(x,y)\in\N$ such that the arc along $f^k(\gamma^{s})$ (respectively, $f^k(\gamma^u)$) connecting the points $f^k(x), f^k(y)$ (respectively, $f^k(F(x)), f^k(F(y))$) belongs to $U$, and 
\begin{align}\label{prelim_bound}
\left[\frac{\dist_{f^{k}(\gamma^s)}\left(f^k(x),f^k(y)\right)}{\dist_{f^{k}(\gamma^u)}\left(f^k(F(x)), f^k(F(y))\right)}\right]^{\pm 1}\leq 2,
\end{align}
where $\dist_{\beta}(a,b)$ denotes the distance between points $a$ and $b$ along the curve $\beta$ (the number $2$ is not significant; anything larger than $1$ will work). Hence it is enough to provide an estimate, independent of $k$, for
\begin{align}\label{prelim_bound2}
\left[\frac{\dist_{\gamma^s}(x,y)}{\dist_{f^{k}(\gamma^s)}\left(f^{k}(x), f^{k}(y)\right)}\right]^{\pm1}\left[\frac{\dist_{f^{k}(\gamma^u)}\left(f^k(F(x)), f^k(F(y))\right)}{\dist_{\gamma^u}\left(F(x), F(y)\right)}\right]^{\pm 1}.
\end{align}
In order to estimate \eqref{prelim_bound2}, it is enough to estimate
\begin{align*}
\frac{\norm{Df^k|_{\gamma^s(x,y)}}}{\norm{Df^k|_{\gamma^u(F(x),F(y))}}},
\end{align*}
where $\gamma^{s,u}(a,b)$ is the arc along $\gamma^{s,u}$ with endpoints $a, b$. After taking $\log$, one estimates the latter by estimating
\begin{align*}
\sum_{j=0}^k\abs{\norm{Df^j|_{\gamma^s(x,y)}} - \norm{Df^j|_{\gamma^u(F(x),F(y))}}}.
\end{align*}
The sum above is majorized by a geometric series, and hence admits an upper bound $L$ for all $k$. One shows that the bound in \eqref{prelim_bound} and the bound $L$, for $L$ sufficiently large, hold for all $f_\alpha$, $\alpha\in(0,\beta)$, with $\beta$ sufficiently small (this follows from continuous dependence of $f_\alpha$ and $Df_\alpha$ on $\alpha$).

Finally, small $C^1$ perturbations of $\tau$ do not destroy these bounds.


\bibliographystyle{plain}
\bibliography{bibliography}

\end{document}